\documentclass[journal,comsoc]{IEEEtran}
\usepackage[T1]{fontenc}

\ifCLASSINFOpdf
\else
\fi
\usepackage{epsfig}
\usepackage{amsthm,amsfonts}
\usepackage{amsmath,amssymb,setspace,cite,color,graphicx}
\usepackage[caption=false,font=footnotesize]{subfig} 
\usepackage{mathtools}

\usepackage{algorithm}
\usepackage[noend]{algorithmic}
\usepackage{booktabs}

\allowdisplaybreaks[1]

\newtheorem{remark}{Remark}

\newtheorem{theorem}{Theorem}
\newtheorem{lemma}{Lemma}

\newtheorem{definition}{Definition}

\newcommand{\byz}{BREA} 

\DeclarePairedDelimiter{\floor}{\lfloor}{\rfloor}

\newcommand{\RNum}[1]{\uppercase\expandafter{\romannumeral #1\relax}}

\begin{document}






\title{Byzantine-Resilient Secure Federated Learning}

\author{Jinhyun So, Ba\c{s}ak~G{\"u}ler, A. Salman Avestimehr
\thanks{Jinhyun So is with the Department of Electrical and Computer Engineering, University of Southern California, Los Angeles, CA, 90089 USA (e-mail: jinhyuns@usc.edu). Ba\c{s}ak G{\"u}ler is with the Department of Electrical and Computer Engineering, University of California, Riverside, CA, 92521 USA (email: bguler@ece.ucr.edu). 
A. Salman Avestimehr is with the Department of Electrical and Computer Engineering, University of Southern California, Los Angeles, CA, 90089 USA (e-mail: avestimehr@ee.usc.edu).

This work is published in Journal on Selected Areas in Communications \cite{so2020byzantine}.
}
}

\maketitle

\begin{abstract}
Secure federated learning is a privacy-preserving framework to improve machine learning models by training over large volumes of data collected by mobile users. This is achieved through an iterative process where, at each iteration, users update a global model using their local datasets. Each user then masks its local update via random keys, and the masked models are aggregated at a central server to compute the global model for the next iteration. As the local updates are protected by random masks, the server cannot observe their true values. This presents a major challenge for the resilience of the model against adversarial (Byzantine) users, who can manipulate the global model by modifying their local updates or datasets. Towards addressing this challenge, this paper presents the first single-server Byzantine-resilient secure aggregation framework (BREA) for secure federated learning. BREA is based on an integrated stochastic quantization, verifiable outlier detection, and secure model aggregation approach to guarantee Byzantine-resilience, privacy, and convergence simultaneously. We provide theoretical convergence and privacy guarantees and characterize the fundamental trade-offs in terms of the network size, user dropouts, and privacy protection. Our experiments demonstrate convergence in the presence of Byzantine users, and comparable accuracy to conventional federated learning benchmarks.

\end{abstract}


\singlespacing 

\section{Introduction}\label{Sec:Intro}
Federated learning is a distributed training framework that has received significant interest in the recent years, by allowing machine learning models to be trained over the vast amount of data collected by mobile devices \cite{pmlr-v54-mcmahan17a, bonawitz2017practical}. 
In this framework, training is coordinated by a central server who maintains a global model, which is updated by the mobile users through an iterative process. At each iteration, the server sends the current version of the global model to the mobile devices, who update it using their local data and create a local update. The server then aggregates the local updates of the users and updates the global model for the next iteration \cite{pmlr-v54-mcmahan17a, bonawitz2017practical, mcmahan2016communication, bonawitz2019towards,  kairouz2019advances,  so2020turbo, li2019federated, yang2018applied}. 

Security and privacy considerations of distributed learning are mainly focused around two seemingly separate directions: 1) ensuring robustness of the global model against adversarial manipulations and 2) protecting the privacy of individual users. The first direction aims at ensuring that the trained model is robust against Byzantine faults that may occur in the training data or during protocol execution. These faults may result either from an adversarial user who can manipulate the training data or the information exchanged during the protocol, or due to device malfunctioning. Notably, it has been shown that even a single Byzantine fault can significantly alter the trained model \cite{blanchard2017machine}. 
The primary approach for defending against Byzantine faults is by comparing the local updates received from different users and removing the outliers at the server \cite{blanchard2017machine, chen2017distributed, pmlr-v80-yin18a, alistarh2018byzantine, yang2019byzantine}. Doing so, however, requires the server to learn the true values of the local updates of each individual user. 
The second direction aims at protecting the privacy of the individual users, by keeping each local update private from the server and the other users participating in the protocol \cite{bonawitz2017practical, mcmahan2016communication, bonawitz2019towards,  kairouz2019advances, so2020turbo, li2019federated}. 
This is achieved through what is known as a \emph{secure aggregation} protocol \cite{bonawitz2017practical}. In this protocol, each user masks its local update through additive secret sharing using private and pairwise random keys before sending it to the server. Once the masked models are aggregated at the server, the additional randomness cancels out and the server learns the aggregate of all user models. At the end of the protocol, the server learns no information about the individual models beyond the aggregated model, as they are masked by the random keys unknown to the server. 
In contrast, conventional distributed training frameworks that perform gradient aggregation and model updates using the true values of the gradients may reveal extensive information about the local datasets of the users, as shown in \cite{NIPS2019_9617, 8737416, geiping2020inverting}.

This presents a major challenge in developing a Byzantine-resilient, and at the same time, privacy-preserving federated learning framework. 
On the one hand, robustness against Byzantine faults requires the server to obtain the individual model updates in the clear, to be able to compare the updates from different users with each other and remove the outliers. 
On the other hand, protecting user privacy requires each individual model to be masked with random keys, as a result, the server only observes the masked model, which appears as a uniformly random vector that could correspond to any point in the parameter space.  
Our goal is to reconcile these two critical directions. In particular, we want to address the following question,
\emph{``How can one make federated learning protocols robust against Byzantine adversaries while preserving the privacy of individual users?''.}

In this paper, we propose the first single-server Byzantine-resilient secure aggregation framework, {\byz}, towards addressing this problem. 
Our framework is built on the following main principles. 
Given a network of $N$ mobile users with up to $A$ adversaries, each user initially secret shares its local update with the other users, through a verifiable secret sharing protocol \cite{feldman1987practical}. However, doing so requires the local updates to be masked by uniformly random vectors in a finite field \cite{shamir1979share}, whereas the model updates during training are performed in the domain of real numbers. In order to handle this problem, {\byz} utilizes stochastic quantization to transfer the local updates from the real domain into a finite field. 


Verifiable secret sharing allows the users to perform consistency checks to validate the secret shares and ensure that every user follows the protocol. 
However, a malicious user can still manipulate the global model by modifying its local update or private dataset. {\byz} handles such attacks through a robust gradient descent approach, enabled by secure computations over the secret shares of the local updates. To do so, each user locally computes the pairwise distances between the secret shares of the local updates belonging to other users, and sends the computation results to the server. 
Since these computations are carried out using the secret shares, users do not learn the true values of the local updates belonging to other users.  

In the final phase, the server collects the computation results from a sufficient number of users, recovers the pairwise distances between the local updates, and performs user selection for model aggregation. The user selection protocol is based on a distance-based outlier removal mechanism \cite{blanchard2017machine}, to remove the effect of potential adversaries and to ensure that the selected models are sufficiently close to an unbiased gradient estimator. 
After the user selection phase, the secret shares of the models belonging to the selected users are aggregated locally by the mobile users. The server then gathers the secure computation results from the users, reconstructs the true value of the aggregate of the selected user models, and updates the global model. Our framework guarantees the privacy of individual user models, in particular, the server learns no information about the local updates, beyond their aggregated value and the pairwise distances. 

In our theoretical analysis, we demonstrate provable convergence guarantees for the model and robustness guarantees against Byzantine adversaries. We then identify the theoretical performance limits in terms of the fundamental trade-offs between the network size, user dropouts, number of adversaries, and  privacy protection. Our results demonstrate that, in a network with $N$ mobile users, {\byz} can theoretically guarantee: i) robustness of the trained model against up to $A$ Byzantine adversaries, ii) tolerance against up to $D$ user dropouts, iii) privacy of each local update, against the server and up to $T$ colluding users, as long as $N\geq 2A+1+\max\{m+2,D+2T\}$, where $m$ is the number of selected models for aggregation.


We then numerically evaluate the performance of {\byz} and compare it to the conventional federated learning protocol, the federated averaging scheme of~\cite{pmlr-v54-mcmahan17a}. To do so, we implement {\byz} in a distributed network of $N=40$ users with up to $A=12$ Byzantine users who can send arbitrary vectors to the server or to the honest users. 
We demonstrate that {\byz} guarantees convergence against Byzantine users and its convergence rate is comparable to the convergence rate of federated averaging. {\byz} also has comparable test accuracy to the federated averaging scheme while {\byz} entails quantization loss to preserve the privacy of individual users.

Finally, while {\byz} provides the first distributed training framework that is both Byzantine-resilient and privacy-preserving, there are several directions for further improvements. The first is that our theoretical analysis for convergence needs the assumption of an independent and identically distributed (i.i.d.) data distribution over the users. Second, we rely on distance-based outlier removal mechanisms, which make it difficult to distinguish whether a large distance between local updates is due to a non-i.i.d. data distribution or a Byzantine attack. Providing theoretical performance guarantees for non-i.i.d. data in the Byzantine-robust training setups is an important and interesting direction for future research. However, one should also note that even without the privacy requirements, this is still an open problem and an active area of research in the literature ~\cite{li2019rsa,data2020byzantine}. Other open questions include whether it is possible to break the quadratic communication and computation complexity of Byzantine-resilient secure aggregation in the current work, and investigating the fundamental performance limits of Byzantine-resilient distributed learning with multiple local updates.


\section{Related Work}
In the non-Byzantine federated learning setting, secure aggregation is performed through a procedure known as additive masking \cite{bonawitz2017practical}, \cite{acs2011have}. In this setup, users first agree on pairwise secret keys using a Diffie-Hellman type key exchange protocol \cite{diffie1976new}, and send a masked version of their local update to the server, where the masking is done using pairwise and private secret keys. 
When the masked models are aggregated at the server, additive masks cancel out, allowing the server to learn the aggregate of the local updates. 
This process works well if no users drop during the execution of the protocol. In wireless environments, however, users may drop from the protocol anytime due to the variations in channel conditions. 
Such user dropouts are handled by letting each user secret share their private and pairwise keys through Shamir's secret sharing \cite{shamir1979share}. The server can remove the additive masks by collecting the secret shares from the surviving users. This approach leads to a quadratic communication overhead in the number of users. 
More recent approaches have focused on reducing the communication overhead, by training in a smaller parameter space \cite{konevcny2016federated}, autotuning the parameters\cite{bonawitz2019federated}, or by utilizing coding techniques \cite{so2020turbo}. 

Another line of work has focused on differentially-private federated learning approaches \cite{geyer2017differentially, sun2019can}, to  protect the privacy of personally-identifiable information against inference attacks. Although our focus is not on differential-privacy, our approach may in principle be combined with differential privacy techniques \cite{dwork2008differential}, which is an interesting future direction. 
Another important direction in federated learning is the study of fairness  and how to avoid biasing the model towards specific users \cite{mohri2019agnostic, li2019fair}. The convergence properties of federated learning models are investigated in \cite{li2019convergence}.   

Distributed training protocols have been extensively studied in the Byzantine setting using clear (unmasked) model updates \cite{blanchard2017machine, chen2017distributed, pmlr-v80-yin18a, alistarh2018byzantine, yang2019byzantine,li2019rsa,data2020byzantine}. The main defense mechanism to protect the trained model against Byzantine users is by comparing the model updates received from different users, and removing the outliers. Doing so ensures that the selected model updates are close to each other, as long as the network has a sufficiently large number of honest users.  
A related line of work is model poisoning attacks, which are studied in \cite{bhagoji2018analyzing, fang2019local}. 

In concurrent work, a Byzantine-robust secure gradient descent algorithm has been proposed for a two-server model in \cite{he2020secure}, however, unlike federated learning (which is based on a single-server architecture)   \cite{pmlr-v54-mcmahan17a, bonawitz2017practical}, this work requires two honest (non-colluding) servers who both interact with the mobile users and communicate with each other to carry out a secure two-party protocol, but do not share any sensitive information with each other in an attempt to breach user privacy. In contrast, our goal is to develop a single-server Byzantine-resilient secure training framework, to facilitate robust and privacy-preserving training architectures for federated learning. Compared to the two-server models, single server models carry the additional challenge where all information has to be collected at a single server, while still being able to keep the individual models of the users private.


The remainder of the paper is organized as follows. 
In Section~\ref{Sec:Background}, we provide background on federated learning. Section~\ref{Sec:System} presents our system model along with the key parameters used to evaluate the system performance. Section~\ref{Sec:framework} introduces our framework and the details of the specific components. Section~\ref{Sec:theoretical_result} presents our theoretical results, whereas our numerical evaluations are provided in Section~\ref{Sec:experiments}, to demonstrate the convergence and Byzantine-resilience. The paper is concluded in Section~\ref{Sec:conclusion}.  
The following notation is used throughout the paper. We represent a scalar variable with $x$, whereas $\mathbf{x}$ represents a vector. A set is denoted by $\mathcal{X}$, whereas $[N]$ refers to the set $\{1, \ldots, N\}$. 


\section{Background} \label{Sec:Background}

\begin{figure}[t]
\centering
\includegraphics[width=\linewidth]{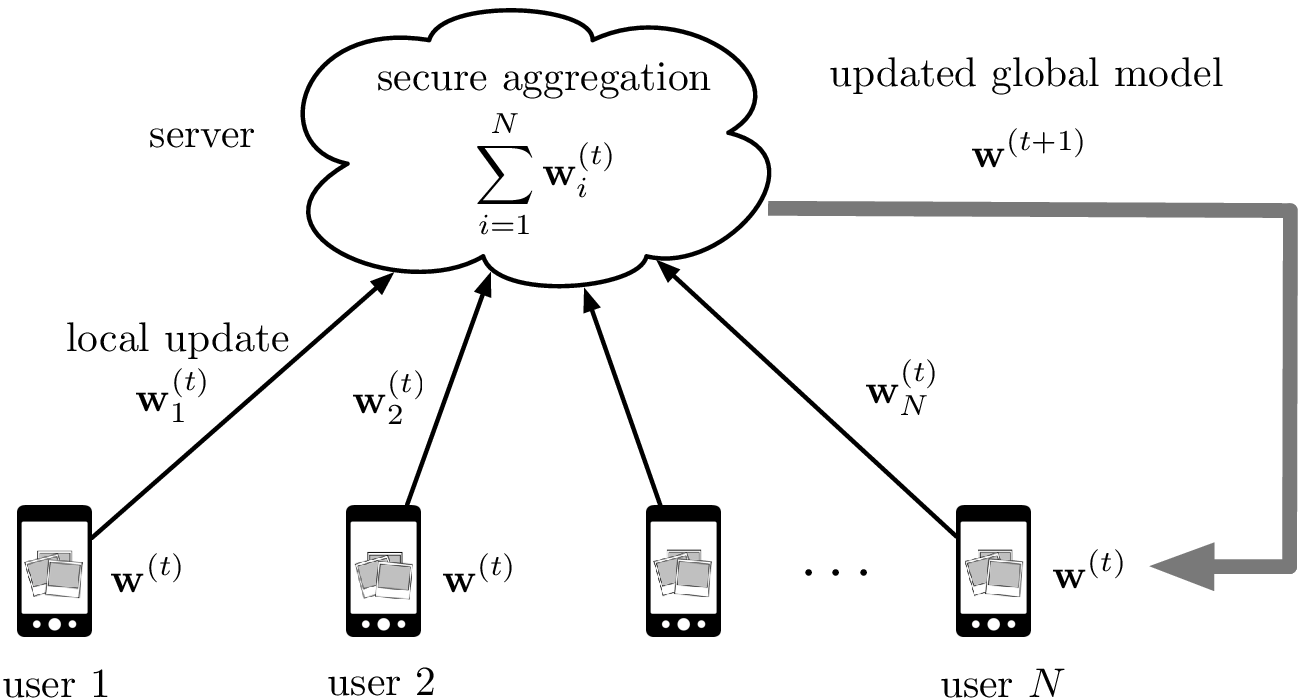}
\caption{Secure aggregation in federated learning. At iteration $t$, the server sends the current state of the global model, denoted by $\mathbf{w}^{(t)}$, to the mobile users. User $i\in[N]$ forms a local update $\mathbf{w}_i^{(t)}$ by updating the global model using its local dataset. The local updates are aggregated in a privacy-preserving protocol at the server, who then updates the global model, and sends the new model, $\mathbf{w}^{(t+1)}$, to the mobile users.}
\label{federated-original}
\vspace{-0.3cm}
\end{figure}

Federated learning is a distributed training framework for machine learning in mobile networks while preserving the privacy of users. 
Training is coordinated by a central server who maintains a global model $\mathbf{w}\in\mathbb{R}^d$ with dimension $d$. 
The goal is to train the global model using the data held at mobile devices, by minimizing a global objective function $C(\mathbf{w})$ as, 
\begin{equation}
\min_{\mathbf{w}} C(\mathbf{w}).  
\end{equation}
The global model is updated locally by mobile users on sensitive private datasets, by letting 
\begin{equation}\label{eq:costfunction}
C(\mathbf{w}) = \sum_{i=1}^N \frac{B_i}{B} C_i(\mathbf{w})  
\end{equation}
where $N$ is the total number of mobile users, $C_i(\mathbf{w})$ denotes the local objective function of user $i$, $B_i$ is the number of data points in user $i$'s private dataset $\mathcal{D}_i$, and $B:=\sum_{i=1}^N B_i$. 
For simplicity, we assume that users have equal-sized datasets, i.e., $B_i = \frac{B}{N}$ for all $i\in[N]$. 

Training is performed through an iterative process where mobile users interact through the central server to update the global model. 
At each iteration, the server shares the current state of the global model, denoted by $\mathbf{w}^{(t)}$, with the mobile users. 
Each user $i$ creates a local update, 
\begin{equation}\label{eq:local_model}
\mathbf{w}_i^{(t)} = g(\mathbf{w}^{(t)},\xi^{(t)}_i)
\end{equation}
where $g$ is an estimate of the gradient $\nabla C(\mathbf{w}^{(t)})$ of the cost function $C$ and $\xi^{(t)}_i$ is a random variable representing the random sample (or a mini-batch of samples) drawn from $\mathcal{D}_i$. We assume that the private datasets $\{\mathcal{D}_i\}_{i\in[N]}$ have the same distribution and $\{\xi^{(t)}_i\}_{i\in[N]}$ are i.i.d. $\xi^{(t)}_i\sim \xi$ where $\xi$ is a uniform random variable such that each  $\mathbf{w}_i^{(t)}$ is an unbiased estimator of the true gradient $\nabla C(\mathbf{w}^{(t)})$, i.e., 
\begin{equation}
\mathbb{E}_\xi [g(\mathbf{w}^{(t)},\xi^{(t)}_i)]=\nabla C(\mathbf{w}^{(t)}).
\end{equation}
The local updates are aggregated at the server in a privacy-preserving protocol, such that the server only learns the aggregate of a large fraction of the local updates, ideally the sum of all user models $\sum_{i\in [N]} \mathbf{w}_i^{(t)}$, 
but no further information is revealed about the individual models  beyond their aggregated value. 
Using the aggregate of the local updates, the server updates the global model for the next iteration,
\begin{equation}
    \mathbf{w}^{(t+1)} = \mathbf{w}^{(t)} - 
    \gamma^{(t)}\sum_{i\in [N]} \mathbf{w}_i^{(t)}
\end{equation}
where $\gamma^{(t)}$ is the learning rate,
and sends the updated model $\mathbf{w}^{(t+1)}$ to the users. 
This process is illustrated in  Figure~\ref{federated-original}.

Conventional secure aggregation protocols require each user to mask its local update using random keys  before aggregation \cite{bonawitz2017practical, so2020turbo, bellsecure}. 
This is typically done by creating pairwise keys between the users through a key exchange protocol \cite{diffie1976new}. Using the pairwise keys, each pair of users $i,j\in[N]$ agree on a pairwise random seed $a_{ij}^{(t)}$. User $i$ also creates a private random seed $b_i^{(t)}$, which protects the privacy of the local update in case the user is delayed instead of being dropped, in which case the pairwise keys are not sufficient for privacy, as shown in \cite{bonawitz2017practical}. 
User $i\in[N]$ then sends a masked version of its local update $\mathbf{w}_i^{(t)}$, given by
\begin{equation}\label{eq:google_modelmasking}
    \mathbf{y}_i^{(t)} := \mathbf{w}_i^{(t)} + \text{PRG}(b_{i}^{(t)})+\sum_{j:i<j}\text{PRG}(a_{ij}^{(t)}) -  \sum_{j:i>j}\text{PRG}(a_{ji}^{(t)})
\end{equation}
to the server, where $\text{PRG}$ is a pseudo random generator. 
User $i$ then secret shares $b_i^{(t)}$ and $\{a_{ij}^{(t)}\}_{j\in[N]}$ with the other users, via Shamir's secret sharing \cite{shamir1979share}. 
For computing the aggregate of the user models, the server collects either the secret shares of the  pairwise seeds belonging to a dropped user, or the shares of the private seed belonging to a surviving user (but not both). 
The server then recovers the private seeds of the surviving users and the pairwise seeds of the dropped users, and removes them from the aggregate of the masked models,
\begin{align} \label{eq:secure_aggregation}
\mathbf{y}^{(t)} &= \sum_{i\in\mathcal{U}}\big(\mathbf{y}_i^{(t)} - \text{PRG}(b_i^{(t)})\big)\notag \\
&\hspace{1cm}-\sum_{i\in\mathcal{D}}\Big(\sum_{j:i<j}\text{PRG}(a_{ij}^{(t)}) - \sum_{j:i>j}\text{PRG}(a_{ji}^{(t)})\Big) \notag \\
&= \sum_{i\in\mathcal{U}}\mathbf{w}_i^{(t)}
\end{align} 
and obtains the aggregate of the local updates, where $\mathcal{U}\subseteq [N]$ and $\mathcal{D}\subseteq [N]$ denote the set of surviving and dropped users, respectively. In~\eqref{eq:secure_aggregation}, $\sum_{i\in\mathcal{U}} \text{PRG}(b_i^{(t)})$ corresponds to the reconstructed private seeds belonging to the surviving users. On the other hand,  $\sum_{i\in\mathcal{D}}\Big(\sum_{j:i<j}\text{PRG}(a_{ij}^{(t)}) - \sum_{j:i>j}\text{PRG}(a_{ji}^{(t)})\Big)$ corresponds to the reconstructed pairwise seeds belonging to the dropped users. Both of these terms are reconstructed by the server to remove the random masks in the aggregate of the masked versions of the surviving users, $\sum_{i\in\mathcal{U}}\mathbf{y}_i^{(t)}$. At the end, all of the random masks cancel out, and the server recovers the summation of the original models belonging to all of the surviving users, i.e., $\sum_{i\in\mathcal{U}}\mathbf{w}_i^{(t)}$.

\section{Problem Formulation} \label{Sec:System}

In this section, we describe the  Byzantine-resilient secure aggregation problem, by extending the conventional secure aggregation scenario from Section~\ref{Sec:Background} to the case when some users, known as Byzantine adversaries, can manipulate the trained model by modifying their local datasets or by sharing false information during the protocol.

We consider a distributed network with $N$ mobile users and a single server. User $i\in[N]$ holds a local update\footnote{For notational clarity, throughout Sections~\ref{Sec:System} and \ref{Sec:framework}, we omit the iteration number $(t)$ from $\mathbf{w}_i^{(t)}$.} $\mathbf{w}_i$ of dimension $d$. 
The goal is to aggregate the local updates at the server, while protecting the privacy of individual users. 
However, unlike the non-Byzantine setting of Section~\ref{Sec:Background}, the aggregation operation in the Byzantine setting should be robust against potentially malicious users. 
To this end, we represent the aggregation operation by a function,
\begin{equation}\label{eq:agg-byz}
f(\mathbf{w}_{1},\ldots,\mathbf{w}_{N}) = \sum_{i\in\mathcal{S}}  \mathbf{w}_i
\end{equation}
where $\mathcal{S}$ is a set of users selected by the server for aggregation. The role of $\mathcal{S}$ is to remove the effect of potentially Byzantine adversaries on the trained model, by removing the outliers. 
Similar to prior works on federated learning, our focus is on computationally-bounded parties, whose strategies can be described by a probabilistic polynomial time algorithm \cite{bonawitz2017practical}.

We evaluate the performance of a Byzantine-resilient secure aggregation protocol according to the following key parameters:
\begin{itemize}

\item \emph{Robustness against Byzantine users:}
We assume that up to $A$ users are Byzantine (malicious), who manipulate the protocol by modifying their local datasets or by sharing false information during protocol execution. The protocol should be robust against such Byzantine adversaries. 

\item \emph{Privacy of local updates:} 
The aggregation protocol should protect the privacy of any  individual user from the server and any collusions between up to $T$ users. Specifically, the local update of any user should not be revealed to the server or the remaining users, even if up to $T$ users cooperate with each other by sharing information.\footnote{Collusions that may occur between the server and the users are beyond the scope of our paper. 
}  


\item \emph{Tolerance to user dropouts:} Due to potentially poor wireless channel conditions, we assume that up to $D$  users may get dropped or delayed at any time during protocol execution. The protocol should be able to tolerate such dropouts, i.e., the privacy and convergence guarantees should hold even if up to $D$ users drop or get delayed.


\end{itemize}



In this paper, we present a single-server Byzantine-resilient secure aggregation framework ({\byz}) for the computation of \eqref{eq:agg-byz}. 
{\byz} consists of the following key components:
\begin{enumerate}

\item \emph{Stochastic quantization:} Users initially quantize their local updates from the real  domain to the domain of integers, and then embed them in a field $\mathbb{F}_{p}$ of integers modulo a prime ${p}$. To do so, our framework utilizes stochastic quantization, which is instrumental in our theoretical convergence guarantees.

\item \emph{Verifiable secret sharing of the user models:}  
Users then secret share their quantized models using a verifiable secret sharing protocol. 
This ensures that the secret shares created by the mobile users are valid, i.e., Byzantine users cannot cheat by sending invalid secret shares.


\item \emph{Secure distance computation:} In this phase, users compute the pairwise distances between the secret shares of the local updates, and send the results to the server. Since this computation is performed using the secret shares of the models instead of their true values, users do not learn any information about the actual model parameters.   

\item \emph{User selection at the server:} 
Upon receiving the computation results from the users, the server recovers the pairwise distances between the local updates and selects the set of users whose models will be included in the aggregation, by removing the outliers. This ensures that the aggregated model is robust against potential manipulations from Byzantine users. The server then announces the list of the selected users.  

\item \emph{Secure model aggregation:} 
In the final phase, each user locally aggregates the secret shares of the models selected by the server, and sends the computation result to the server. Using the computation results, the server recovers the aggregate of the models of the selected users, and updates the model.

\end{enumerate}

In the following, we describe the details of each phase. 




\section{The {\byz} Framework} \label{Sec:framework}
In this section, we present the details of the {\byz} framework for Byzantine-resilient secure federated learning. 
\vspace{-0.2cm}
\subsection{Stochastic Quantization}\label{sec:quantization} 
In {\byz}, the operations for verifiable secret sharing and secure distance computations are carried out over a finite field $\mathbb{F}_{p}$ for some large prime ${p}$. To this end, user $i\in[N]$ initially quantizes  its local update $\mathbf{w}_i$ from the domain of real numbers to the finite field. 
We assume that the field size ${p}$ is large enough to avoid any wrap-around during secure distance computation and secure model aggregation, which will be described in Sections~\ref{sec:secure_dist} and \ref{sec:secure_aggr}, respectively. 

Quantization requires a challenging task as it should be performed in a way to ensure the convergence of the model. 
Moreover, the quantization function should allow the representation of negative integers in the finite field, and facilitate computations to be performed in the quantized domain. Therefore, we cannot utilize well-known gradient quantization techniques such as in \cite{alistarh2017qsgd}, which represents the sign of a negative number separately from its magnitude.  
{\byz} addresses this challenge with a simple stochastic quantization strategy as follows. For any integer $q\geq1$, we define a stochastic rounding function:
\begin{equation}\label{eq:sto_round}
    Q_q(x) = 
    \left\{
    \begin{array}{ll}
          \frac{\lfloor qx \rfloor}{q}   & \text{with prob. } 1-(qx-\lfloor qx \rfloor)\\
          \frac{\lfloor qx \rfloor+1}{q} & \text{with prob. } qx-\lfloor qx \rfloor
    \end{array} 
    \right. 
\end{equation}
where $\floor{x}$ is the largest integer less than or equal to $x$, and note that this function is unbiased, i.e., $\mathbb{E}_Q[Q_q(x)]=x$. Parameter $q$ is a tuning parameter that corresponds to the number of quantization levels. Variance of $Q_q(x)$ decreases as the value of $q$ increases, which will be detailed in Lemma~\ref{lemma:q_prop} in Section~\ref{Sec:theoretical_result}. 
We then define the quantized model, 
\begin{equation}\label{eq:def_w_bar}
    \overline{\mathbf{w}}_i := \phi(q\cdot Q_q({\mathbf{w}}_i))
\end{equation}
where the function $Q_q$ from~\eqref{eq:sto_round} is carried out element-wise, and a mapping function $\phi:\mathbb{R}\rightarrow\mathbb{F}_p$ is defined to represent a negative integer in the finite field by using two's complement representation, 
\begin{equation}\label{eq:phi} 
    \phi(x) =
    \left\{
    \begin{array}{ll}
          x & \text{if } x \geq 0\\
          p+x & \text{if } x<0.
    \end{array} 
    \right. 
\end{equation}



\vspace{-0.4cm}
\subsection{Verifiable Secret Sharing of the User Models}\label{Sec:VSS} 
{\byz} protects the privacy of individual user models through verifiable secret sharing. 
This is to ensure that the individual user models are kept private while preventing the Byzantine users from breaking the integrity of the protocol by sending invalid secret shares to the other users.

To do so, user $i\in[N]$ secret shares its quantized model $\overline{\mathbf{w}}_i$ with the other users through a non-interactive verifiable secret sharing protocol~\cite{feldman1987practical}.
Our framework leverages Feldman's verifiable secret sharing protocol from~\cite{feldman1987practical}, which combines Shamir's secret sharing~\cite{shamir1979share} with homomorphic encryption. In this setup, each party creates the secret shares, then broadcasts commitments to the coefficients of the polynomial they use for Shamir's secret sharing, so that other parties can verify that the secret shares are constructed correctly. To verify the secret shares from the given commitments, the protocol leverages the homomorphic property of exponentiation, i.e., $exp(a+b) = exp(a) exp(b)$, whereas the privacy protection is based on the assumption that computation of the discrete logarithm in the finite field is intractable. 

%
%
%
%
The individual steps carried out for verifiable secret sharing in our framework are as follows. 
Initially, the server and users agree on $N$ distinct elements $\{\theta_i\}_{i\in[N]}$ from $\mathbb{F}_{p}$. This can be done offline by using a conventional majority-based consensus protocol \cite{dolev1983authenticated, coulouris2005distributed}.
User $i\in[N]$ then generates secret shares of the quantized model $\overline{\mathbf{w}}_i$ by forming a random  polynomial $f_i\!:\!\mathbb{F}_{p}\!\rightarrow\!\mathbb{F}_{p}^d$ of degree ${T}$,
\begin{equation}\label{eq:VSS_gen}
    f_i(\theta) = \overline{\mathbf{w}}_i + \sum_{j=1}^{T} \mathbf{r}_{ij}\theta^{j}
\end{equation}
in which the vectors $\mathbf{r}_{ij}$ are generated uniformly at random from $\mathbb{F}^{d}_{p}$ by user $i$. 
User $i$ then sends a secret share of $\overline{\mathbf{w}}_i$ to user $j$, denoted by,
\begin{equation}\label{eq:VSS_def}
    \mathbf{s}_{ij}=f_i(\theta_j).
\end{equation}
To make these shares verifiable, user $i$ also  broadcasts commitments to the coefficients of $f_i$, given by
\begin{equation}\label{eq:def:commitment}
    \mathbf{c}_{ij} = 
    \left\{
    \begin{array}{ll}
          \psi^{\overline{\mathbf{w}}_i} & \text{for } j=0\\
          \psi^{\mathbf{r}_{ij}} & \text{for } j=1,\ldots,{T}.
    \end{array} 
    \right.
\end{equation}
where $\psi$ denotes a generator of $\mathbb{F}_p$, and all arithmetic is taken modulo $\lambda$ for some large prime $\lambda$ such that $p$ divides $\lambda-1$.

Upon receiving the commitments in \eqref{eq:def:commitment}, each user $j\in[N]$ can verify the secret share $\mathbf{s}_{ij}=f_i(\theta_j)$ by checking 
\begin{equation}\label{eq:VSStest}
    \psi^{\mathbf{s}_{ij}} = \prod_{k=0}^{T} \mathbf{c}_{ik} ^{\theta_j^k}
\end{equation}
where all arithmetic is taken modulo $\lambda$. This commitment scheme ensures that the secret shares are created correctly from the polynomial in~\eqref{eq:VSS_gen}, hence they are valid. 
On the other hand, as we assume the intractability of computing the discrete logarithm~\cite{feldman1987practical}, the server or the  users cannot compute the discrete logarithm $\log_{\psi}(\mathbf{c}_{it})$ and reveal the quantized model $\overline{\mathbf{w}}_i$ from $\mathbf{c}_{i0}$ in~\eqref{eq:def:commitment}. 






\vspace{-0.3cm}
\subsection{Secure Distance Computation}\label{sec:secure_dist}
Verifiable secret sharing of the model parameters, as described in Section~\ref{Sec:VSS}, ensures that the users follow the protocol correctly by creating valid secret shares. However, malicious users can still try to manipulate the trained model by modifying their local updates instead. In this case, the secret shares will be created correctly but according to a false model. 
In order to ensure that the trained model is robust against such adversarial manipulations, {\byz} leverages a distance-based outlier detection mechanism, such as in~\cite{blanchard2017byzantine,blanchard2017machine}. The main principle behind these mechanisms is to compute the pairwise distances between the local updates and select a set of models that are sufficiently close to each other. 
On the other hand, the outlier detection mechanism in {\byz} has to protect the privacy of local updates, and performing the distance computations on the true values of the model parameters would breach the privacy of individual users. 

We address this by a privacy-preserving distance computation approach, in which the pairwise distances are computed locally by each user, using the secret shares of the model parameters received from the other users.  
In particular, upon receiving the secret shares of the model parameters as described in Section~\ref{Sec:VSS}, user $i$ computes the pairwise distances,
\begin{equation}\label{eq:dist}
d^{(i)}_{jk} := \lVert \mathbf{s}_{ji} - \mathbf{s}_{ki} \rVert^2
\end{equation}
between each pair of users $j,k\in[N]$, and sends the result to the server. Since the computations in \eqref{eq:dist} are performed over the secret shares, user $i$ learns no information about the true values of the model parameters $\overline{\mathbf{w}}_{j}$ and $\overline{\mathbf{w}}_{k}$ of users $j$ and $k$, respectively. 
Finally, we note that the computation results from \eqref{eq:dist} are scalar values.  

\vspace{-0.2cm}
\subsection{User Selection at the Server}\label{sec:selection}
Upon receiving the computation results in  \eqref{eq:dist} from a sufficient number of users, the server reconstructs the true values of the pairwise distances. During this phase, Byzantine users may send incorrect computation results to the server, hence the reconstruction process should be able to correct the potential errors that may occur in the computation results due to malicious users. Our decoding procedure is based on the decoding of Reed-Solomon codes. 

The intuition of the decoding process is that the computations from~\eqref{eq:dist} correspond to evaluation points of a univariate polynomial $h_{jk}\!:\!\mathbb{F}_{p}\!\rightarrow\!\mathbb{F}_{p}$ of degree at most $2T$, where
\begin{equation}\label{eq:poly_dist}
    h_{jk}(\theta) := \lVert f_j(\theta) - f_k(\theta) \rVert^2
\end{equation}
for $\theta\in\{\theta_i\}_{i\in[N]}$ and $j,k\in[N]$.
Accordingly, $h_{jk}$ can be viewed as the encoding polynomial of a Reed-Solomon code with degree at most $2{T}$, such that the missing computations due to the dropped users correspond to the \emph{erasures} in the code, and manipulated computations from Byzantine users refer to the  \emph{errors} in the code. Therefore, the decoding process of the server corresponds to decoding an $[N,2{T}+1,N-2{T}]_p$ Reed-Solomon code with at most $D$ erasures and at most $A$ errors. By utilizing well-known Reed-Solomon decoding algorithms~\cite{gao2003new}, the server can recover the polynomial $h_{jk}$ and obtain the true value of the pairwise distances by using the relation  $h_{jk}(0)=\lVert f_j(0) - f_k(0) \rVert^2=\lVert \overline{\mathbf{w}}_j - \overline{\mathbf{w}}_k\rVert^2$.
At the end, the server learns the pairwise distances 
\begin{equation}\label{eq:distance_finite_field}
\overline{d}_{jk} := \lVert \overline{\mathbf{w}}_j - \overline{\mathbf{w}}_k\rVert^2
\end{equation}
between the models of users $j,k\in[N]$. 
Then the server converts \eqref{eq:distance_finite_field} from the finite field to the real domain as follows, 
\begin{equation}\label{eq:distance_real_domain}
    {d}_{jk} = \frac{{\phi}^{-1} (\overline{d}_{jk})}{q^2}
\end{equation}
for $j,k\in[N]$, where $q$ is the integer parameter in~\eqref{eq:sto_round} and the demapping function $\phi^{-1}:\mathbb{F}_p\rightarrow\mathbb{R}$ is defined as 
\begin{equation}\label{eq:inv_phi}
{\phi}^{-1}(\overline{x})=
    \left\{
    \begin{array}{ll}
          \overline{x} & \text{if \quad } 0 \leq \overline{x} < \frac{p-1}{2}\\
          \overline{x}-p & \text{if \quad } \frac{p-1}{2} \leq \overline{x} < p
    \end{array} .
    \right.
\end{equation}
We assume the field size $p$ is large enough to ensure the correct recovery of the pairwise distances, 
\begin{align}
    {d}_{jk}
    &=\frac{\phi^{-1}\big( \lVert \phi(q\cdot Q_q(\mathbf{w}_j)) - \phi(q\cdot Q_q(\mathbf{w}_k))\rVert^2 \big)}{q^2} \notag \\
    &=\frac{\phi^{-1}\big( \phi( q^2 \lVert Q_q(\mathbf{w}_j) - Q_q(\mathbf{w}_k) \rVert^2 ) \big)}{q^2} \label{eq:cond_fieldsize} \\
    &=\lVert Q_q(\mathbf{w}_j) - Q_q(\mathbf{w}_k) \rVert^2 \label{eq:distance_equality}    
\end{align}
where $Q_q$ is the stochastic rounding function defined in~\eqref{eq:sto_round} and \eqref{eq:cond_fieldsize} holds if
\begin{equation}\label{eq:cond2_fieldsize}
    q^2\lVert Q_q(\mathbf{w}_j) - Q_q(\mathbf{w}_k)\rVert^2<(p-1)/2.    
\end{equation}

By utilizing the pairwise distances in~\eqref{eq:distance_equality}, the server carries out a distance-based outlier removal algorithm to select the set of users to include in the final model aggregation. The outlier removal procedure of {\byz} follows the multi-Krum algorithm from~\cite{blanchard2017byzantine,blanchard2017machine}. 
The main difference is that our framework considers the multi-Krum algorithm in a quantized stochastic gradient setting, as {\byz} utilizes quantized gradients instead of the true gradients, in order to enable privacy-preserving computations on the secret shares. 
We present the theoretical convergence guarantees of this quantized multi-Krum algorithm in Section~\ref{Sec:theoretical_result}, and numerically demonstrate its convergence behaviour in our experiments in  Section~\ref{Sec:experiments}.

In this setup, the server selects $m$ users through the following iterative process.
At each iteration $k\in[m]$, the server selects one user, denoted by $i^{(k)}$, by finding 
\begin{equation}\label{eq:mKrum_iter}
i^{(k)} = \arg \min_{j\in[N]\setminus\mathcal{S}^{(k-1)}} s^{(k)}(j)
\end{equation}
where $\mathcal{S}^{(k)}$ denotes the index set of the users selected in up to $k$ iterations and $s^{(k)}(j)$ is a score function assigned to user $j$ at iteration $k$. 
The score function of user $j$ is defined as 
\begin{equation}
s^{(k)}(j) = \sum_{u\in\mathcal{I}^{(k-1)}_j} d_{ju}
\end{equation}
where $\mathcal{I}^{(k-1)}_j\subseteq[N]\setminus\mathcal{S}^{(k-1)}$ denotes the set of $(N-k+1)-A-2$ users whose models are closest to the model of user $j$. After selecting $i^{(k)}$, the server updates the selected index set as $\mathcal{S}^{(k)}=\{\mathcal{S}^{(k-1)}, i^{(k)} \}$ where $\mathcal{S}^{(0)}=\emptyset$. After $m$ iterations, the server obtains the index set $\mathcal{S}=\mathcal{S}^{(m)}$. 

\vspace{-0.2cm}
\subsection{Secure Model Aggregation}\label{sec:secure_aggr}
The final phase of {\byz} is to securely aggregate the local updates of the selected users, without revealing the individual models to the server. 
To do so, the server initially announces the list of selected users via broadcasting. We denote the set of selected users by $\mathcal{S}$. 
Then, each user locally aggregates the secret shares belonging to the selected users,
\begin{equation}\label{eq:SS_sum}
    \mathbf{s}_i = \sum_{j\in\mathcal{S}} \mathbf{s}_{ji}
\end{equation}
and sends the result to the server. Upon receiving the computation results from a sufficient number of users, the server can decode the aggregate of the models $\sum_{j\in\mathcal{S}} \overline{\mathbf{w}}_j$ through the decoding of Reed-Solomon codes.

The intuition of the decoding process is similar to the decoding of the pairwise distances described in Section~\ref{sec:selection}. Specifically, the computations from~\eqref{eq:SS_sum} can be viewed as evaluation points of a univariate polynomial $h:\!:\!\mathbb{F}_{p}\!\rightarrow\!\mathbb{F}_{p}^d$ of degree at most ${T}$,
\begin{equation}\label{eq:poly_aggr}
    h(\theta) := \sum_{j\in\mathcal{S}} f_j(\theta).
\end{equation}
One can then observe that $h$ is the encoding polynomial of a Reed-Solomon code with degree at most $T$, the missing computations due to the dropped users correspond to the erasures in the code, and manipulated computations from Byzantine users correspond to the errors in the code. Therefore, the decoding process of the server corresponds to decoding an $[N,{T}+1,N-{T}]_p$ Reed-Solomon code with at most $D$ erasures and at most $A$ errors. Hence, by using a Reed-Solomon decoding algorithm, the server can recover the polynomial $h$ and obtain the true value of the aggregate of the selected user models by using the relation  $h(0)=\sum_{j\in\mathcal{S}} f_j(0)=\sum_{j\in\mathcal{S}} \overline{\mathbf{w}}_j$. 
We note that the total number of users selected by the server for aggregation, i.e., $|\mathcal{S}| = m$, should be sufficiently large, which can be agreed offline between the users and the server. Then, if the set announced by the server is too  small (e.g., consisting of a single user), the honest users may opt to not send the computation results.

Upon learning the aggregate of the user models, the server updates the global model for the next iteration as follows,
\begin{equation}\label{eq:up1}
    \mathbf{w}^{(t+1)} = \mathbf{w}^{(t)} - \frac{\gamma^{(t)}}{q} \phi^{-1}\Big( \sum_{j\in\mathcal{S}} \overline{\mathbf{w}}_j^{(t)} \Big)
\end{equation}
where $\phi^{-1}$ is the demapping function defined in~\eqref{eq:inv_phi} and $q$ is the integer parameter in~\eqref{eq:sto_round}. 
We assume that the field size $p$ is large enough to avoid wrap-around in $\sum_{j\in\mathcal{S}} \overline{\mathbf{w}}_j^{(t)}$ such that 
\begin{align}
    \sum_{j\in\mathcal{S}} \overline{\mathbf{w}}_j^{(t)}
    &=\sum_{j\in\mathcal{S}}\phi(q \cdot Q_q(\mathbf{w}^{(t)}_j)) \label{eq:quantized_w} \\
    &=\phi\big(q\sum_{j\in\mathcal{S}}\mathbf{w}^{(t)}_j\big) \label{eq:cond_final_eq}
\end{align}
where~\eqref{eq:quantized_w} follows from~\eqref{eq:def_w_bar}.
Finally, it follows from \eqref{eq:cond_final_eq} that the update equation in \eqref{eq:up1} is equivalent to 
\begin{equation}\label{eq:desired_updateeq}
    \mathbf{w}^{(t+1)} = \mathbf{w}^{(t)} - \gamma^{(t)} \sum_{j\in\mathcal{S}} Q_q( \mathbf{w}^{(t)}_j)
\end{equation}
where $Q_q$ is the stochastic rounding function defined in~\eqref{eq:sto_round}. 

Having all above steps, the overall {\byz} framework can now be presented in Algorithm \ref{Alg}.

\begin{algorithm}[t!]
  \caption{Byzantine-Resilient Secure Aggregation ({\byz})}\label{Alg} 
  \begin{algorithmic}[1]
    \INPUT{Local dataset $\mathcal{D}_i$ of users $i\in[N]$, number of iterations $J$.} \
    \OUTPUT{Global model $(\mathbf{w}^{(J)})$.} \ 
    \vspace{0.1cm} 
    \FOR{iteration $t=0,\ldots,J-1$}
        \FOR{user $i=1,\ldots,N$}
            \STATE Download the global model ${\mathbf{w}}^{(t)}$ from the server.
            
            \STATE Create a local update ${\mathbf{w}}_{i}^{(t)}$ from \eqref{eq:local_model}. 
            
            \STATE Create the quantized model $\overline{\mathbf{w}}_i = \phi(q\cdot Q_q({\mathbf{w}}_i))$ from \eqref{eq:def_w_bar}.
            
            \STATE Generate secret shares $\{\mathbf{s}_{ij}\}_{j\in[N]}$ from \eqref{eq:VSS_def} and send  $\mathbf{s}_{ij}$ to user $j\in[N]$.
            \STATE Generate commitments $\{\mathbf{c}_{ij}\}_{j\in[T]}$ from \eqref{eq:def:commitment} and broadcast $\{\mathbf{c}_{ij}\}_{j\in[T]}$ to all users.
            
            \STATE Verify the secret shares $\{\mathbf{s}_{ji}\}_{j\in[N]}$ by testing \eqref{eq:VSStest}.
            
            \STATE Compute $\{d^{(i)}_{jk}\}_{j,k\in[N]}$ from \eqref{eq:dist} and send the results to the server.
        
            
        \ENDFOR
    
        \STATE Server recovers   $\{\overline{d}_{jk}\}_{j,k\in[N]}$ in \eqref{eq:distance_finite_field} from the computation results  $\{d^{(i)}_{jk}\}_{i,j,k\in[N]}$ by utilizing   Reed-Solomon decoding.
        
        \STATE Server converts $\{\overline{d}_{jk}\}_{j,k\in[N]}$ from the finite field to the real domain to obtain the pairwise distances $\{{d}_{jk}\}_{j,k\in[N]}$ from \eqref{eq:distance_real_domain}.
        
        \STATE Server selects the set $\mathcal{S}$ by utilizing the multi-Krum algorithm \cite{blanchard2017machine} based on the pairwise distances $\{{d}_{jk}\}_{j,k\in[N]}$.
        
        \STATE Server broadcasts the set $\mathcal{S}$ to all users.
    
        \FOR{user $i=1,\ldots,N$}
            \STATE Compute $\mathbf{s}_i = \sum_{j\in\mathcal{S}} \mathbf{s}_{ji}$ from \eqref{eq:SS_sum} and send the result to the server.
        \ENDFOR
        
        \STATE Server recovers $\sum_{i\in\mathcal{S}} \mathbf{w}^{(t)}_{i}$ from the computation results $\{\mathbf{s}_i\}_{i\in[N]}$ by utilizing  Reed-Solomon decoding.
        
        \STATE Server updates the global model, $\mathbf{w}^{(t+1)} = \mathbf{w}^{(t)} - \frac{\gamma^{(t)}}{q} \phi^{-1}\Big( \sum_{i\in\mathcal{S}} \overline{\mathbf{w}}_i^{(t)} \Big)$.
    \ENDFOR
    
    

  \end{algorithmic}
\end{algorithm}





\section{Theoretical Analysis} \label{Sec:theoretical_result}
In this section, we analyze the fundamental performance limits of {\byz}. 
The global model update equation of {\byz} can be expressed as follows,
\begin{equation}\label{eq:update_eq}
    \mathbf{w}^{(t+1)} = \mathbf{w}^{(t)} - \gamma^{(t)} f\big(Q_q(\mathbf{w}^{(t)}_1),\ldots,Q_q(\mathbf{w}^{(t)}_N)\big)
\end{equation} 
where $f$ is the aggregation operation  from~\eqref{eq:agg-byz} and represents the user selection and model aggregation procedures from Sections~\ref{sec:selection} and \ref{sec:secure_aggr}, respectively, while $Q_q$ is the stochastic rounding function defined in~\eqref{eq:sto_round}. 



As described in Section~\ref{Sec:Background}, the local update $\mathbf{w}^{(t)}_i$ created by an honest user is an unbiased estimator of the true  gradient, i.e., $\mathbf{w}^{(t)}_i = g(\mathbf{w}^{(t)},\xi^{(t)}_i)$ with $\mathbb{E}_\xi [g(\mathbf{w}^{(t)},\xi^{(t)}_i)]=\nabla C(\mathbf{w}^{(t)})$ where $\xi_i^{(t)}\sim\xi$ and $\xi$ is a uniform random variable representing the random sample (or a mini-batch of samples) drawn from the dataset. 
We define the local standard deviation $\sigma$ of the gradient estimator $g$ by
\begin{equation}
d\sigma^2(\mathbf{w}) := \mathbb{E}_\xi \lVert g(\mathbf{w},\xi) - \nabla C(\mathbf{w}) \rVert^2
\end{equation}
for all $\mathbf{w}\in \mathbb{R}^d$.
The model created by a Byzantine user  can refer to any random vector $\mathbf{b}_i^{(t)}\in\mathbb{R}^d$, which we represent as $\mathbf{w}^{(t)}_i =\mathbf{b}_i^{(t)}$. Accordingly, the quantized model $Q(\mathbf{w}^{(t)}_i)$ of Byzantine user could refer to any vector in $\mathbb{Z}^d$. 


Our first lemma states the unbiasedness and bounded variance of the quantized gradient estimator $Q_q(g(\mathbf{w},\xi))$ for any vector $\mathbf{w}\in\mathbb{R}^d$.
\begin{lemma}\label{lemma:q_prop}
    For the quantized gradient estimator $Q_q(g(\mathbf{w},\xi))$ with a given vector $\mathbf{w}\in\mathbb{R}^d$
    where $\xi$ is a uniform random variable representing the sample drawn, $g$ is a gradient estimator such that $\mathbb{E}_\xi [g(\mathbf{w},\xi)]=\nabla C(\mathbf{w})$ and $\mathbb{E}_\xi \lVert g(\mathbf{w},\xi) - \nabla C(\mathbf{w}) \rVert^2 = d\sigma^2(\mathbf{w})$, and the stochastic rounding function $Q_q$ is given in~\eqref{eq:sto_round}, the following holds,
    \begin{align}
        \mathbb{E}_{Q,\xi} [Q_q(g(\mathbf{w},\xi))] &= \nabla C(\mathbf{w}) \label{unbiased}\\
        \mathbb{E}_{Q,\xi} \lVert Q_q(g(\mathbf{w},\xi)) - \nabla C(\mathbf{w}) \rVert^2
        &\leq d\sigma ^{\prime \;2}(\mathbf{w}) \label{variance}
    \end{align}
    where $\sigma ^{\prime}(\mathbf{w})=\sqrt{\frac{1}{4q^2} + \sigma^2(\mathbf{w})}$.
\end{lemma}

\begin{proof}
    (Unbiasedness) Given $Q_q$ in~\eqref{eq:sto_round} and any random variable $x$, it follows that,
    \begin{align}
        \mathbb{E}_Q [Q_q(x) \mid x] 
        =&\; \frac{\lfloor qx \rfloor}{q} (1-(qx-\lfloor qx \rfloor)) \notag\\
        &+ \frac{(\lfloor qx \rfloor+1)}{q} (qx-\lfloor qx \rfloor) \notag \\
        =&\; x 
    \end{align}
    from which we obtain the unbiasedness condition in \eqref{unbiased},
    \begin{align}
        \mathbb{E}_{Q,\xi} [Q_q(g(\mathbf{w},\xi))] 
        &= \mathbb{E}_{\xi}\big[ \mathbb{E}_{Q}[Q_q(g(\mathbf{w},\xi))\mid g(\mathbf{w},\xi)] \big] \notag \\
        &= \mathbb{E}_{\xi}\big[ g(\mathbf{w},\xi) \big] \notag \\
        &= \nabla C(\mathbf{w}). 
    \end{align}

\noindent    
    (Bounded variance)
    Next, we observe that,
    \begin{align}
        &\mathbb{E}_Q \Big[ \big(Q_q(x) - \mathbb{E}_Q[Q_q(x)\mid x]\big)^2 \mid x \Big] \notag \\
        &\quad = \Big(\frac{\lfloor qx \rfloor}{q} - x\Big)^2 (1-(qx-\lfloor qx \rfloor)) \notag\\
        &\quad \;\;+ \Big(\frac{\lfloor qx \rfloor+1}{q}-x \Big)^2(qx-\lfloor qx \rfloor) \notag \\
        &\quad = \frac{1}{q^2}\Big(\frac{1}{4} - \big(qx-\lfloor qx \rfloor -\frac{1}{2}\big)^2 \Big) \notag \\
        &\quad \leq \frac{1}{4q^2} \label{eq:lem1_ineq1}
    \end{align}
from which one can obtain the bounded variance condition in \eqref{variance} as follows,
    \begin{align}
        &\mathbb{E}_{Q,\xi} \lVert Q_q(g(\mathbf{w},\xi)) - \nabla C(\mathbf{w}) \rVert^2 \notag \\
        &\quad = \mathbb{E}_{\xi}\big[ \mathbb{E}_{Q}[ \lVert Q_q(g(\mathbf{w},\xi)) - \nabla C(\mathbf{w}) \rVert^2 \mid g(\mathbf{w},\xi)] \big] \notag \\
        &\quad \leq \mathbb{E}_{\xi}\big[ \mathbb{E}_{Q}[ \lVert Q_q(g(\mathbf{w},\xi)) - g(\mathbf{w},\xi) \rVert^2 \mid g(\mathbf{w},\xi)] \big] \notag \\
        &\quad \qquad + \mathbb{E}_{\xi}\big[ \mathbb{E}_{Q}[ \lVert g(\mathbf{w},\xi) - \nabla C(\mathbf{w}) \rVert^2 \mid g(\mathbf{w},\xi)] \big] \label{eq:triangle_ineq} \\
        &\quad \leq \frac{d}{4q^2} + d \sigma^2(\mathbf{w})  \label{eq:lem1_ineq2} \\
        &\quad = d \sigma ^{\prime\;2}(\mathbf{w}) \notag
    \end{align}
    where \eqref{eq:triangle_ineq} follows from the triangle inequality and \eqref{eq:lem1_ineq2} follows form~\eqref{eq:lem1_ineq1}.
\end{proof}
As discussed in Section~\ref{Sec:System}, Byzantine users can manipulate the training protocol via two means, either by modifying their local update (directly or by modifying the local dataset), or by sharing false information during protocol execution. In this section, we demonstrate how {\byz} provides robustness in both cases. 
We first focus on the former case and study the resilience of the global model, i.e., conditions under which the trained model remains close to the true model, even if some users modify their local updates adversarially. 
The second case, i.e., robustness of the protocol when some users exchange false information during the protocol execution, will be considered in Theorem~\ref{thm:convergence}. 

In order to evaluate the resilience of the global model against Byzantine adversaries, we adopt the notion of $(\alpha, A)$-Byzantine resilience from \cite{blanchard2017byzantine}. 
\begin{definition}[$(\alpha, A)$-Byzantine resilience, \cite{blanchard2017byzantine}]\label{def:byzantine_resilience}
    Let $0\leq\alpha<\pi/2$ be any angular value and $0\leq A\leq N$ be any integer.
    Let $\mathbf{w}_1,\ldots,\mathbf{w}_{N}\in\mathbb{R}^d$ be any  i.i.d random vectors such that $\mathbf{w}_i\sim\mathbf{w}$ with $\mathbb{E}[\mathbf{w}] = \mathbf{W}$.
    Let $\mathbf{b}_1,\ldots,\mathbf{b}_{A}\in\mathbb{R}^d$ be any random vectors.
    Then, function $f$ is $(\alpha, A)$-Byzantine resilient if, for any $1\leq j_1<\cdots<j_A\leq N$, 
    \begin{equation}
    \mathbf{f}:=f\big(\mathbf{w}_1,\ldots,\underbrace{\mathbf{b}_1}_{j_1},\ldots,\underbrace{\mathbf{b}_A}_{j_A},\ldots,\mathbf{w}_{N}\big)
    \end{equation}
    satisfies, i) $\mathbf{W}^{\top}\mathbb{E}[\mathbf{f}]\geq(1-\sin\alpha)\lVert \mathbf{W} \rVert^2$, and 
    ii) for $r\in\{2,3,4\}$, $\mathbb{E}\lVert \mathbf{f}\rVert^{r}$ is bounded above by 
    $\mathbb{E}\lVert \mathbf{f}\rVert^{r} \leq K \sum_{r_1+\cdots+r_{N-A}=r}\mathbb{E}\lVert \mathbf{W}\rVert^{r_1}\ldots\mathbb{E}\lVert \mathbf{W}\rVert^{r_{N-A}}$ where $K$ denotes a generic constant.
\end{definition}

Lemma~\ref{lemma:Byzantine_resilience} below states that if the standard deviation caused by random sample selection and quantization is smaller than the norm of the true gradient, and $2A+2 < N-m$, then the aggregation function $f$ from \eqref{eq:update_eq} is $(\alpha,A)$-Byzantine resilient where $\alpha$ depends on the ratio of the standard deviation over the norm of the gradient~\cite{blanchard2017byzantine}. 
\begin{lemma}\label{lemma:Byzantine_resilience}
    Assume that $2A+2 < N-m$ and $\eta(N,A)\sqrt{d}\sigma^{\prime} < \lVert \nabla C(\mathbf{w}) \rVert$ where 
    \begin{equation}
        \eta(N,A) := \sqrt{2\Big( N-A + \frac{A(N-A-2)+A^2(N-A-1)}{N-2A-2}\Big)}.  
    \end{equation}
Let $\mathbf{w}_1,\ldots,\mathbf{w}_N$ be i.i.d. random vectors in $\mathbb{R}^d$ such that $\mathbf{w}_i \sim \mathbf{w}$ with $\mathbb{E}_{\xi}[g(\mathbf{w},\xi)]=\nabla C(\mathbf{w})$ and $\mathbb{E}_\xi \lVert g(\mathbf{w},\xi) - \nabla C(\mathbf{w}) \rVert^2 = d\sigma^2(\mathbf{w})$. Then, the aggregation function $f$ from \eqref{eq:update_eq} is $(\alpha,A)$-Byzantine resilient where $0\leq\alpha<\pi/2$ is defined by $\sin{\alpha} = \frac{\eta(N,A)\sqrt{d}\sigma^{\prime}}{\lVert \nabla C(\mathbf{w}) \rVert}$. 
\end{lemma}
    
\begin{proof} 
    From Lemma~\ref{lemma:q_prop}, $\mathbb{E}_{Q,\xi} [Q_q(g(\mathbf{w},\xi))] = \nabla C(\mathbf{w})$ and $\mathbb{E}_{Q,\xi} \lVert Q_q(g(\mathbf{w},\xi)) - \nabla C(\mathbf{w}) \rVert^2 \leq d\sigma ^{\prime \;2}(\mathbf{w})$. 
    Then, the  quantized multi-Krum algorithm described in  Section~\ref{sec:selection}, where the multi-Krum algorithm applied to the quantized vectors $Q_q(\mathbf{w}_i)$, is $(\alpha,A)$-Byzantine resilient from Proposition $3$ of~\cite{blanchard2017byzantine}. Hence, function $f$ in \eqref{eq:update_eq} is $(\alpha,A)$-Byzantine resilient.     
\end{proof}


We now state our main result for the theoretical performance guarantees of {\byz}. 

\begin{theorem}\label{thm:convergence}
    We assume that:  
    1) the cost function $C$ is three times differentiable with continuous derivatives, and is bounded from below, i.e., $C(x)\geq0$; 
    2) the learning rates satisfy, $\sum_{t=1}^{\infty} \gamma^{(t)} = \infty$ and $\sum_{t=1}^{\infty} ({\gamma^{(t)}})^2 < \infty$;
    3) the second, third, and fourth moments of the quantized gradient estimator do not grow too fast with the norm of the model, i.e., $\forall r\in{2,3,4}$, $\mathbb{E}_{Q,\xi}\lVert Q_q(g(\mathbf{w},\xi))\rVert^r\leq A_r + B_r\lVert \mathbf{w} \rVert^r$ for some constants $A_r$ and $B_r$;
    4) there exist a constant $0\leq \alpha <\pi/2$ such that for all $\mathbf{w}\in\mathbb{R}^d$, $\eta(N,A)\sqrt{d}\sigma^{\prime}(\mathbf{w}) \leq \lVert \nabla C(\mathbf{w}) \rVert \sin{\alpha}$;
    5) the gradient of the cost function $C$ satisfies that for $\lVert \mathbf{w} \rVert^2\geq R$, there exist constants $\epsilon>0$ and $0\leq\beta<\pi/2-\alpha$ such that 
    \begin{align}
        \lVert \nabla C(\mathbf{w}) \rVert \geq \epsilon > 0,\\
        \frac{\mathbf{w}^{\top} C(\mathbf{w})}{\lVert \mathbf{w} \rVert \cdot \lVert \nabla C(\mathbf{w})\rVert} \geq \cos{\beta}.
    \end{align}
    Then, {\byz} guarantees,
    \begin{itemize}
        \item (Robustness against Byzantine users) The protocol executes correctly against up to $A$ Byzantine users and the trained model is $(\alpha,A)$-Byzantine resilient.
        
        \item (Convergence) The sequence of the gradients $\nabla C(\mathbf{w}^{(t)})$ converges almost surely to zero,
        \begin{equation}
        \nabla C(\mathbf{w}^{(t)})\xrightarrow[t \to \infty]{a.s.} 0. 
        \end{equation}
        
        \item (Privacy) 
        The server or any group of up to $T$ users cannot compute an unknown local update. 
        For any set of users  $\mathcal{T}\subset[N]$ of size at most $T$,
        \begin{align}
            &\mathbb{P}[\text{User $i$ has secret $\overline{\mathbf{w}}_i$} \mid \text{view}_{\mathcal{T}}] \notag \\ 
            &\hspace{2cm}= \mathbb{P}[\text{User $i$ has secret $\overline{\mathbf{w}}_i$}] \label{eq:privacy}
        \end{align}
        for all $i\in[N]\setminus\mathcal{T}$ where $\text{view}_{\mathcal{T}}$ denotes the messages that the members of $\mathcal{T}$ receive.
    \end{itemize}
    for any $N\geq 2A+1+\max\{m+2,D+2T\}$, where $m$ is the number of selected models for aggregation.
\end{theorem}

\begin{remark} The two conditions $\sum_{t=1}^{\infty} \gamma^{(t)} = \infty$ and $\sum_{t=1}^{\infty} ({\gamma^{(t)}})^2 < \infty$ are instrumental in the convergence of stochastic gradient descent algorithms \cite{bottou1998online}. Condition $\sum_{t=1}^{\infty} ({\gamma^{(t)}})^2 < \infty$ states that the learning rates decrease fast enough, whereas  condition $\sum_{t=1}^{\infty} \gamma^{(t)} = \infty$  bounds the rate of their decrease, to ensure that the learning rates do not decrease too fast. 
\end{remark}

\begin{remark}
We consider a general (possibly non-convex) objective function $C$.  In such scenarios, proving the convergence of the model directly is challenging, and various approaches have been proposed instead. Our approach follows \cite{bottou1998online} and \cite{blanchard2017machine}, where we prove the convergence of the gradient to a flat region instead. We note, however, that such a region may refer to any stationary point, including the local minima as well as saddle and extremal points.    
\end{remark}


\begin{proof}

(Robustness against Byzantine users) 
The $(\alpha,A)$-Byzantine resilience of the trained model follows from Lemma~\ref{lemma:Byzantine_resilience}. 
We next provide sufficient conditions for {\byz} to correctly evaluate the update function  \eqref{eq:update_eq}, in the presence of $A$ Byzantine users.
Byzantine users may send any arbitrary random vector to the server or other users in every step of the protocol in Section~\ref{Sec:framework}. 
In particular, Byzantine users can create and send incorrect computations in three attack scenarios: i) sending invalid secret shares $\mathbf{s}_{ij}$ in~\eqref{eq:VSS_def}, ii) sending incorrect secure distance computations $d_{jk}^{(i)}$ in~\eqref{eq:dist}, and iii) sending incorrect aggregate of the secret shares $\mathbf{s}_i$ in~\eqref{eq:SS_sum}.
    
    
    
The first attack scenario occurs when the secret shares
$\mathbf{s}_{ij}$ in~\eqref{eq:VSS_def} do not refer to the same polynomial from \eqref{eq:VSS_gen}. 
{\byz} utilizes verifiable secret sharing to prevent such attempts. The correctness (validity) of the secret shares can be verified by testing~\eqref{eq:VSStest}, whenever the majority of the surviving users are honest, i.e., $N>2A+D$~\cite{feldman1987practical,dolev1983authenticated}.

The second attack scenario can be detected and corrected by the Reed-Solomon decoding algorithm. In particular, as described in Section~\ref{sec:selection}, given $j,k\in[N]$, $\{d_{jk}^{(i)}\}_{i\in[N]}$ can be viewed as $N$ evaluation points of the  polynomial $h_{jk}$ given in~\eqref{eq:poly_dist} whose degree is at most $2T$. The decoding process at the server then corresponds to the decoding of an  $[N,2T+1,N-2T]_p$ Reed-Solomon code with at most $D$ erasures and at most $A$ errors. As an $[n,k,n-k+1]_p$ Reed-Solomon code with $e$ erasures can tolerate a maximum number of $\lfloor \frac{n-k-e}{2}\rfloor$  errors~\cite{gao2003new}, the server can recover the correct pairwise distances as long as  $A\leq \lfloor \frac{N-(2T+1)-D}{2}\rfloor$, i.e. $N\geq D+2A+2T+1$. 

The third attack scenario can also be detected and corrected by the Reed-Solomon decoding algorithm. As described in Section~\ref{sec:secure_aggr}, $\{\mathbf{s}_i\}_{i\in[N]}$ are evaluation points of polynomial $h$ in~\eqref{eq:poly_aggr} of degree at most $T$. This decoding process corresponds to the  decoding of an $[N,T+1,N-T]_p$ Reed-Solomon code with at most $D$ erasures and at most $A$ errors. As such, the server can recover the desired aggregate model $h(0)=\sum_{j\in{\mathcal{S}}}\overline{\mathbf{w}}_j$ as long as $N\geq D+2A+T+1$. Therefore, combining with the condition of Lemma~\ref{lemma:Byzantine_resilience}, the sufficient conditions under which {\byz} guarantees robustness against Byzantine users is given by
\begin{equation}\label{eq:sufficient_cond}
    N\geq 2A + 1 + \max\{m+2,D+2T\}.
\end{equation}

(Convergence)
We now consider the update equation in \eqref{eq:update_eq} and prove the convergence of the random sequence $\nabla C(\mathbf{w}^{(t)})$.
From Lemma~\ref{lemma:Byzantine_resilience}, the quantized multi-Krum function $f$ in \eqref{eq:update_eq} is $(\alpha,A)$-Byzantine resilient.
Hence, from Proposition $2$ of~\cite{blanchard2017byzantine}, $\nabla C(\mathbf{w}^{(t)})$ converges almost surely to zero, 
\begin{equation}
    \nabla C(\mathbf{w}^{(t)})\xrightarrow[t \to \infty]{a.s.} 0. 
\end{equation}

(Privacy) 
As described in Section~\ref{Sec:VSS}, we assume the intractability of computing discrete logarithms, hence the server or any user cannot compute $\overline{\mathbf{w}}_i$ from $\mathbf{c}_{i0}$ in~\eqref{eq:def:commitment}.  
It is therefore sufficient to prove the privacy of each individual model against a group of $T$ colluding users, in the case where $\mathcal{T}$ has size $T$. 
If $T$ users cannot get any information about $\overline{\mathbf{w}}_i$, then neither can fewer than $T$ users. 
Without loss of generality, let $\mathcal{T}=\{1,\ldots,T\}=[T]$ and
$\text{view}_{\mathcal{T}} = \big\{ (\mathbf{s}_{kj} )_{k\in[N]}, \big( d^{j}_{kl} \big)_{k,l\in[N]}, \mathbf{s}_j \big\}_{j\in[T]}$ where $\mathbf{s}_{kj}$ in~\eqref{eq:VSS_def} is the secret share of $\overline{\mathbf{w}}_k$ sent from user $k$ to user $j$, $d^{j}_{kl}$ in~\eqref{eq:dist} is the pairwise distance of the secret shares sent from users $k$ and $l$ to user $j$, and $\mathbf{s}_j$ in~\eqref{eq:SS_sum} is the aggregate of the secret shares.
As $\big\{ \big( d^{j}_{kl} \big)_{k,l\in[N]}\big\}_{j\in[T]}$ and $\big\{ \mathbf{s}_j\big\}_{j\in[T]}$ are determined by 
$\big\{ \big( \mathbf{s}_{kj} \big)_{k,l\in[N]}\big\}_{j\in[T]}$, we can simplify the left hand side (LHS) of \eqref{eq:privacy} as
\begin{align}
    &\mathbb{P}[\text{User $i$ has secret $\overline{\mathbf{w}}_i$} \mid \text{view}_{\mathcal{T}}] \notag\\
    &\qquad= \mathbb{P}[\text{User $i$ has secret $\overline{\mathbf{w}}_i$} \mid \{ \mathbf{s}_{kj} \}_{k\in[N],j\in[T]}]\\
    &\qquad= \mathbb{P}[\text{User $i$ has secret $\overline{\mathbf{w}}_i$} \mid \{ \mathbf{s}_{ij} \}_{j\in[T]}]. \label{eq:s_kj_indep}
\end{align}
where \eqref{eq:s_kj_indep} follows from the fact that $\mathbf{s}_{kj}$ is independent of $\overline{\mathbf{w}}_i$ for any $k\neq i$, $\mathbf{s}_{kj}$ is independent of $\overline{\mathbf{w}}_i$.
Then, for any realization of vectors $\mathbf{\rho}_0,\ldots,\mathbf{\rho}_T\in\mathbb{F}_p^d$, we obtain,
\begin{align}
    &\mathbb{P}[\text{User $i$ has secret $\overline{\mathbf{w}}_i$} \mid \text{view}_{\mathcal{T}}] \notag\\
    &\qquad= \mathbb{P}[{\overline{\mathbf{w}}_i = \mathbf{\rho}_0} \mid \mathbf{s}_{i1}= \mathbf{\rho}_1,\ldots,\mathbf{s}_{iT}= \mathbf{\rho}_T] \notag \\
    &\qquad= \frac{\mathbb{P}[\overline{\mathbf{w}}_i = \mathbf{\rho}_0, \mathbf{s}_{i1}= \mathbf{\rho}_1,\ldots,\mathbf{s}_{iT}= \mathbf{\rho}_T]}{\mathbb{P}[\mathbf{s}_{i1}= \mathbf{\rho}_1,\ldots,\mathbf{s}_{iT}= \mathbf{\rho}_T]} \notag \\
    &\qquad= \frac{1/|\mathbb{F}_p^d|^{T+1}}{1/|\mathbb{F}_p^d|^{T}} \label{eq:prob_equality}\\
    &\qquad= \frac{1}{|\mathbb{F}_p^d|} = \mathbb{P}[{\overline{\mathbf{w}}_i = \mathbf{\rho}_0}] \notag
\end{align}
where \eqref{eq:prob_equality} follows from the fact that any $T+1$ evaluation points define a unique polynomial of degree $T$, which completes the proof of privacy. 
\end{proof}



\subsection{Complexity Analysis} \label{Sec:Complexity}
\begin{table}[t!]
\small
\caption{\color{black}Complexity summary of {\byz}. }
\vspace{-0.2cm}
\label{tbl:complexity}
\begin{center}
\begin{tabular}{ccc}
\toprule
 & Computation & Communication   \\
\midrule
 Server  & $O((N^3 + dN)\log^2{N}\log\log{N})$  & $O(dN+N^3)$  \\
 User    & $O(dN\log^2{N}+dN^2)$                & $O(dN+N^2)$  \\
\bottomrule
\end{tabular}
\end{center}
\vspace{-0.3cm}
\end{table}

{\color{black}

\noindent In this section, we analyze the complexity of {\byz} with respect to the number of users, $N$, and model dimension $d$.


\noindent {\bf Complexity Analysis of the Users:}
User $i$'s computation cost can be broken into three parts: 1) generating the secret shares $\mathbf{s}_{ij}$ in~\eqref{eq:VSS_gen} for $j\in[N]$, 2) computing the pairwise distances $d^{(i)}_{jk}$ in~\eqref{eq:dist} for $j,k\in[N]$, and 3) aggregating the secret shares belonging to the selected users from~\eqref{eq:SS_sum}. First, generating $N$ secret shares of a vector with dimension $d$ has a computation cost of $O(dN\log^2 N)$~\cite{shamir1979share}. Second, as there are $O(N^2)$ pairwise distances, computing the pairwise distances has a computation cost of $O(d N^2)$ in total. Third, when the number of selected users $m=|\mathcal{S}|$ is $O(N)$, aggregating the secret shares belonging to the selected users has a computation cost of $O(dN)$. Therefore, the overall computation cost of each user is $O(dN\log^2 N + dN^2)$.

User $i$'s communication cost can be broken to three parts: 1) sending the secret share $\mathbf{s}_{ij}$ to user $j\in[N]$, 2) sending the secret shares of the pairwise distances $d^{(i)}_{jk}$ to the server for $j,k\in[N]$, and 3) sending the aggregate of secret the shares $\mathbf{s}_i$ in~\eqref{eq:SS_sum} to the server. The communication cost of the three parts are $O(dN)$, $O(N^2)$, and $O(d)$, respectively. Therefore, the overall communication cost of each user is $O(dN + N^2)$.

\noindent {\bf Complexity Analysis of the Server:}
Computation cost of the server can be broken into three parts: 1) decoding the pairwise distances by recovering $h_{jk}(0)$ in~\eqref{eq:poly_dist} for $j,k\in[N]$, 2) carrying out the multi-Krum algorithm to select the $m$ users for aggregation, and 3) decoding the aggregate of the selected models by recovering $h(0)$ in~\eqref{eq:poly_aggr}. As described in Section~\ref{sec:selection}, recovering the polynomial $h_{jk}$ corresponds to decoding an $[N,2T+1,N-2T]_p$ Reed-Solomon code with at most $D$ erasures and at most $A$ errors, which has an $O(N\log^2 N \log \log N)$ computation cost~\cite{gao2003new}. As there are $O(N^2)$ pairs and each pair is embedded in a single polynomial, the computation cost of the first part is $O(N^3\log^2 N \log \log N)$ in total. The computation cost of the second part is $O(dN^2 + N^2\log N)$~\cite{blanchard2017machine}. As the dimension of $h(0)$ is $d$, the  computation cost of the third part is $O(dN\log^2 N \log \log N)$. Overall, the computation cost of the server is $O((N^3+dN)\log^2 N \log\log{N})$.

Communication cost can be broken into two parts: 1) receiving the secret shares of the  pairwise distances $d^{(i)}_{jk}$ from users $i\in[N]$ for $j,k\in[N]$ and 2) receiving the aggregate of secret shared models $\mathbf{s}_i$ from users $i\in[N]$. The communication cost of the two parts are $O(N^3)$ and $O(dN)$, respectively. Overall, the  communication cost of the server is $O(dN + N^3)$.

We summarize the complexity analysis in Table~\ref{tbl:complexity}. 
As can be observed in Table~\ref{tbl:complexity}, the server has a communication cost of $O(N^3)$  and a computation cost of  $O(N^3\log^2{N}\log\log{N})$, which is due to the recovery of $N^2$ pairwise distances.
Although the distances are scalar valued, the overhead can become a limitation for very large-scale networks. 
In the next subsection, we propose a generalized framework to reduce the communication cost from $O(N^3)$ to $O(N^2)$ as well as to reduce the computation cost from $O(N^3\log^2{N}\log\log{N})$ to $O(N^2\log^2{N}\log\log{N})$. 
}

\subsection{The Generalized {\byz} Framework} \label{Sec:NewAlgorithm}
The key idea of the generalized framework is to partition the set of $N(N-1)/2$ pairwise distances into sets of size $K$, and embed the $K$ distances in a single polynomial. 
Consequently, the number of polynomials to embed the $N(N-1)/2$ pairwise distances can be reduced from $O(N^2)$ to $O(N^2/K)$.
Each user then sends a single evaluation point of each polynomial to the server, which has an $O(N^3/K)$ communication cost in total. 
By setting $K=O(N)$, the generalized {\byz} framework can achieve $O(N^2)$ communication complexity. 

\begin{table}[t!]
\small
\caption{\color{black}Complexity summary of generalized {\byz}. }
\vspace{-0.2cm}
\label{tbl:complexity_genBrea}
\begin{center}
\begin{tabular}{ccc}
\toprule
 & Computation & Communication   \\
\midrule
 Server  & $O((N^2 + dN)\log^2{N}\log\log{N})$  & $O(dN+N^2)$  \\
 User    & $O(dN\log^2{N}+dN^2)$                & $O(dN+N)$  \\
\bottomrule
\end{tabular}
\end{center}
\vspace{-0.3cm}
\end{table}

We now present the details of the generalized {\byz} framework. 
In the stochastic quantization phase, the generalized {\byz} framework follows the same steps as in Section~\ref{sec:quantization}.
In the verifiable secret sharing phase, user $i\in[N]$ generates secret shares of the quantized model $\overline{\mathbf{w}}_i$ by modifying the random polynomial $f_i$ in~\eqref{eq:VSS_gen} as,
\begin{equation}\label{eq:modified_VSS_gen}
    f_i(\theta) = \overline{\mathbf{w}}_i + \sum_{j=1}^{T} \mathbf{r}_{ij}\theta^{K+j-1}
\end{equation}
where the degree of $f_i$ is increased from $T$ to $K+T-1$.
User $i$ then sends a secret share of $\overline{\mathbf{w}}_i$ to user $j\in[N]$, denoted by $\mathbf{s}_{ij}=f_i(\theta_j)$. 

In the secure distance computation phase, user $i$ computes the pairwise distances $d_{jk}^{(i)}$ as in~\eqref{eq:dist} for $j,k\in[N]$. 
Note that there are $\frac{N(N-1)}{2}$ unique distance values, due to the symmetry of the distance measure. 
The server and the users then agree on a partition $\mathcal{P}_1,\ldots,\mathcal{P}_{G}$ of the $\frac{N(N-1)}{2}$ distances into sets of size $K$, where $G=\lceil \frac{N(N-1)}{2K}\rceil$ and $|\mathcal{P}_k|=K$ for all $k\in[G]$.  $\lceil x\rceil$ denotes the smallest integer greater or equal to $x$. 
Then, each user embeds the $K$ pairwise distances in each $\mathcal{P}_k$ into a single polynomial, for every $k\in\{1, \ldots, G\}$. 
For instance, let the first set in the partition be $\mathcal{P}_1 = \{(1,2),\ldots,(1,K+1)\}$, Then, one can define a univariate polynomial ${u}_{\mathcal{P}_1}:\mathbb{F}_p\rightarrow\mathbb{F}_p$ of degree $2(K+T-1)+K-1$ to embed the $K$ pairwise distances of the pairs in $\mathcal{P}_1$ as
\begin{align}
    {u}_{\mathcal{P}_1}(\theta) &= \lVert \mathbf{f}_{1}(\theta) - \mathbf{f}_{2}(\theta) \rVert^2 + \theta\lVert \mathbf{f}_{1}(\theta) - \mathbf{f}_{3}(\theta) \rVert^2 \notag\\
    &\quad  + \cdots + \theta^{K-1}\lVert \mathbf{f}_{1}(\theta) - \mathbf{f}_{K+1}(\theta) \rVert^2.
    \label{eq:K_embed_poly}
\end{align}
As the coefficients from the first degree to the $(K-1)$-th degree terms in~\eqref{eq:modified_VSS_gen} are zero, the  coefficient of the $(k-1)$-th degree term in~\eqref{eq:K_embed_poly} corresponds to the pairwise distance of the $k$-th pair in $\mathcal{P}_1$, i.e., $\lVert \mathbf{f}_{1}(0) - \mathbf{f}_{k+1}(0) \rVert^2 = \lVert \overline{\mathbf{w}}_{1} - \overline{\mathbf{w}}_{k+1} \rVert^2$ for all $k\in[K]$.
User $i$ then sends  ${u}_{\mathcal{P}_1}(\theta_i)$ to the server.
Upon receiving the computation results from a sufficient number of users, the server can decode the pairwise distances of all pairs in $\mathcal{P}_1$ by reconstructing the polynomial ${u}_{\mathcal{P}_1}(\theta_i)$.
As the degree of ${u}_{\mathcal{P}_1}$ is $3K+2T-3$, the minimum number of results the server needs to collect from the users to recover the pairwise distances, i.e., the recovery threshold, is $3K+2T-2$.
The decoding process of the polynomial corresponds to decoding an $[N, 3K+2T-3, N-3K-2T+2]_p$ Reed-Solomon code with at most $D$ erasures and at most $A$ errors. 
In a similar way, we can define polynomials ${u}_{\mathcal{P}_k}$ for $k\in \{2,\ldots,G\}$ where user $i$ computes and sends ${u}_{\mathcal{P}_k}(\theta_i)$ to the server. 
The server can then decode all of the $N(N-1)/2$ pairwise distances by reconstructing the $G$ polynomials. 

After the server learns the pairwise distances, the generalized protocol follows the same steps as in Sections~\ref{sec:selection} and \ref{sec:secure_aggr}.
The overall algorithm of the generalized {\byz} framework is presented in Algorithm \ref{Alg_gene}.

The generalized {\byz} framework achieves a communication cost of $O(N^2)$  and computation cost of $O(N^2\log^2{N}\log\log{N})$ by setting $K=O(N)$, which follows from the following observations.
First, user $i$ sends an evaluation point of each polynomial $u_{\mathcal{P}_k}$ to the server for $k\in[G]$, and there are $N$ users and $G=O(N^2/K)=O(N)$ polynomials, which has $O(N^2)$ communication overhead in total.
Second, the computation cost to decode the polynomial $u_{\mathcal{P}_k}$ is $O(N\log^2{N}\log\log N)$~\cite{gao2003new} and there are $G=O(N)$ polynomials, which has a computation cost of $O(N^2\log^2{N}\log\log N)$ in total. We summarize the complexity of the generalized {\byz} protocol in Table \ref{tbl:complexity_genBrea}.

\begin{algorithm}[t!]
  \caption{\color{black} Generalized {\byz}}\label{Alg_gene} 
  \begin{algorithmic}[1]
    \INPUT{Local dataset $\mathcal{D}_i$ of users $i\in[N]$, number of iterations $J$.} \
    \OUTPUT{Global model $(\mathbf{w}^{(J)})$.} \ 
    \vspace{0.1cm} 
    \FOR{iteration $t=0,\ldots,J-1$}
        \FOR{user $i=1,\ldots,N$}
            \STATE Download the global model ${\mathbf{w}}^{(t)}$ from the server.
            
            \STATE Create a local update ${\mathbf{w}}_{i}^{(t)}$ from \eqref{eq:local_model}. 
            
            \STATE Create the quantized model $\overline{\mathbf{w}}_i = \phi(q\cdot Q_q({\mathbf{w}}_i))$ from \eqref{eq:def_w_bar}.
            
            \STATE Generate secret shares $\{\mathbf{s}_{ij}=f_i(\theta_j)\}_{j\in[N]}$ from \eqref{eq:modified_VSS_gen} and send  $\mathbf{s}_{ij}$ to user $j\in[N]$.
            
            \STATE Generate commitments $\{\mathbf{c}_{ij}\}_{j\in[T]}$ from \eqref{eq:def:commitment} and broadcast $\{\mathbf{c}_{ij}\}_{j\in[T]}$ to all users.
            
            \STATE Verify the secret shares $\{\mathbf{s}_{ji}\}_{j\in[N]}$ by testing \eqref{eq:VSStest}.
            
            \STATE Compute $\{d^{(i)}_{jk}\}_{j,k\in[N]}$ from \eqref{eq:dist} and send the results to the server.
        
            
        \ENDFOR
    
        \STATE Server recovers $\{\overline{d}_{jk}\}_{j,k\in[N]}$ in \eqref{eq:distance_finite_field} from the coefficients of polynomials ${u}_{\mathcal{P}_1},\ldots,{u}_{\mathcal{P}_G}$ in \eqref{eq:K_embed_poly}
        by utilizing Reed-Solomon decoding.
        
        \STATE Server converts $\{\overline{d}_{jk}\}_{j,k\in[N]}$ from the finite field to the real domain to obtain the pairwise distances $\{{d}_{jk}\}_{j,k\in[N]}$ from \eqref{eq:distance_real_domain}.
        
        \STATE Server selects the set $\mathcal{S}$ by utilizing the multi-Krum algorithm \cite{blanchard2017machine} based on the pairwise distances $\{{d}_{jk}\}_{j,k\in[N]}$.
        
        \STATE Server broadcasts the set $\mathcal{S}$ to all users.
    
        \FOR{user $i=1,\ldots,N$}
            \STATE Compute $\mathbf{s}_i = \sum_{j\in\mathcal{S}} \mathbf{s}_{ji}$ from \eqref{eq:SS_sum} and send the result to the server.
        \ENDFOR
        
        \STATE Server recovers $\sum_{i\in\mathcal{S}} \mathbf{w}^{(t)}_{i}$ from the computation results $\{\mathbf{s}_i\}_{i\in[N]}$ by utilizing  Reed-Solomon decoding.
        
        \STATE Server updates the global model, $\mathbf{w}^{(t+1)} = \mathbf{w}^{(t)} - \frac{\gamma^{(t)}}{q} \phi^{-1}\Big( \sum_{i\in\mathcal{S}} \overline{\mathbf{w}}_i^{(t)} \Big)$.
    \ENDFOR
    
    

  \end{algorithmic}
\end{algorithm}

\section{Experiments} \label{Sec:experiments}
In this section, we demonstrate the convergence and resilience properties of {\byz} compared to conventional federated learning, i.e., the federated averaging scheme from \cite{pmlr-v54-mcmahan17a}, which is termed  \emph{FedAvg} throughout the section. We measure the performance in terms of the cross entropy loss evaluated over the training samples and the model accuracy evaluated over the test samples, with respect to the iteration index, $t$.



\begin{figure}[t]
\centering
\includegraphics[width=\linewidth]{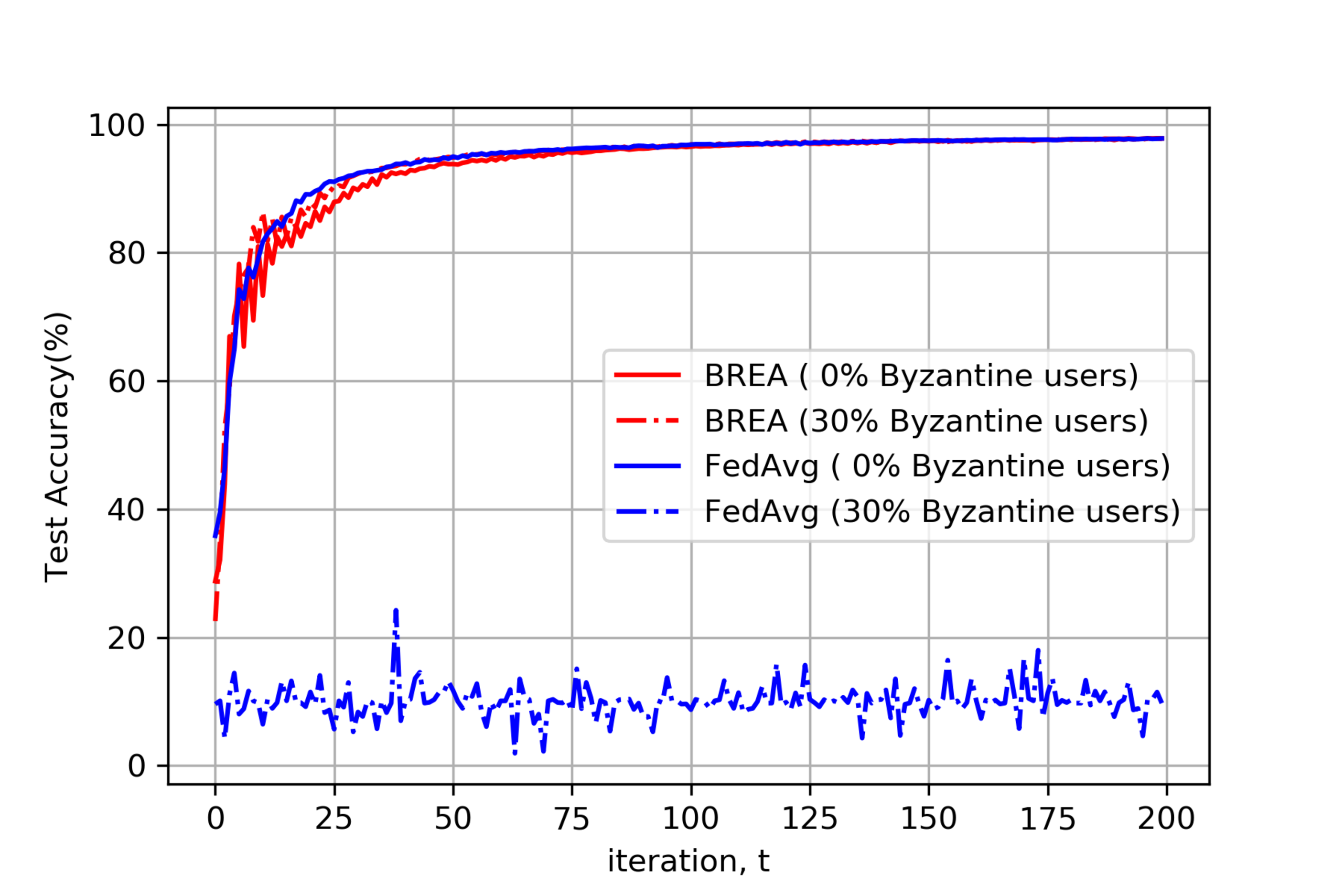}
\caption{Test accuracy of {\byz} and FedAvg~\cite{pmlr-v54-mcmahan17a} for different number of Byzantine users {\color{black}with i.i.d. MNIST dataset}.}
\label{fig:test_accuracy}
\vspace{-0.3cm}
\end{figure}

\begin{figure}[t]
\centering
\includegraphics[width=\linewidth]{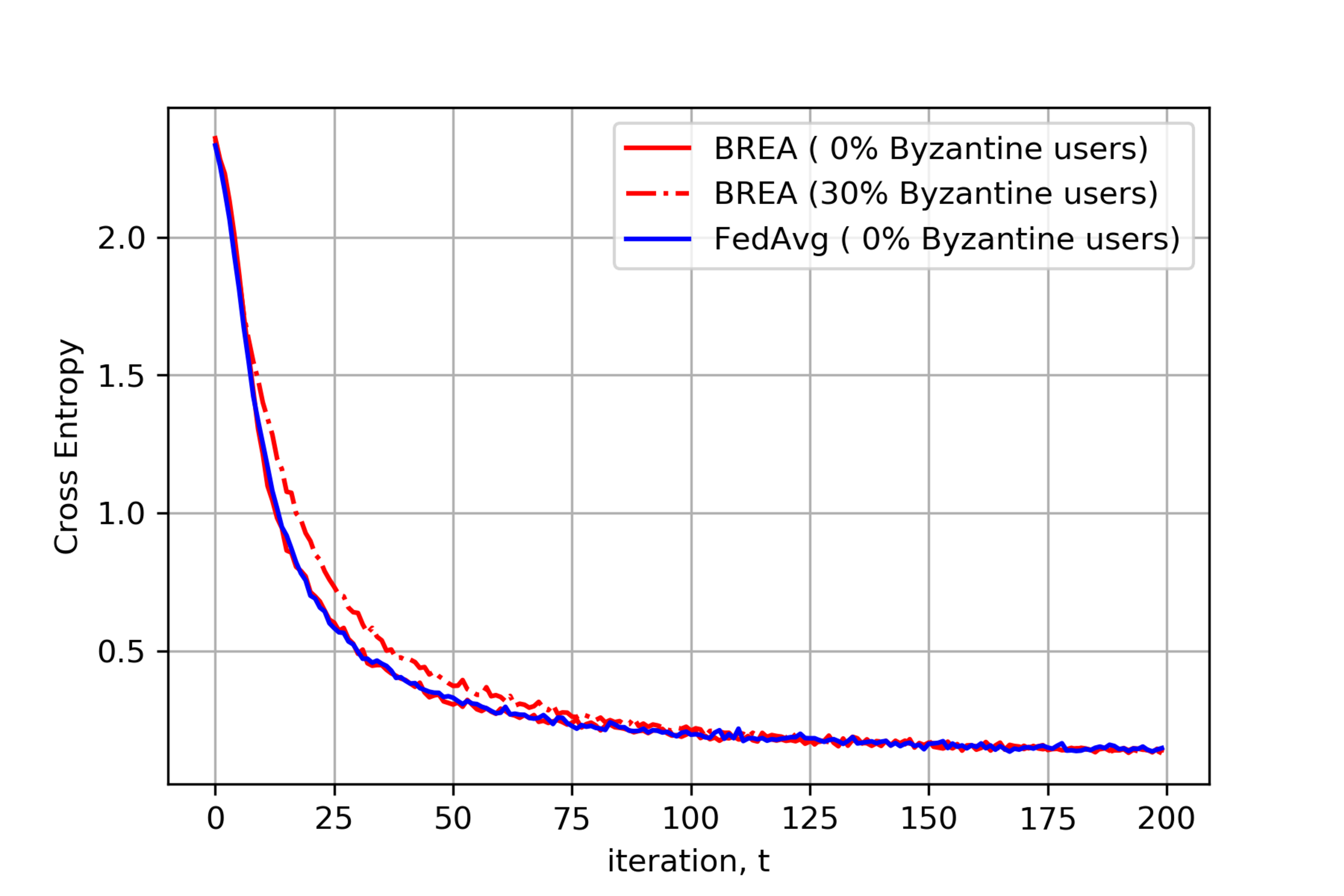}
\caption{Convergence of {\byz} and FedAvg~\cite{pmlr-v54-mcmahan17a} for different number of Byzantine users {\color{black}with i.i.d. MNIST dataset}.}
\label{fig:cross_entropy}
\vspace{-0.3cm}
\end{figure}

\begin{figure}[t]
\centering
\includegraphics[width=\linewidth]{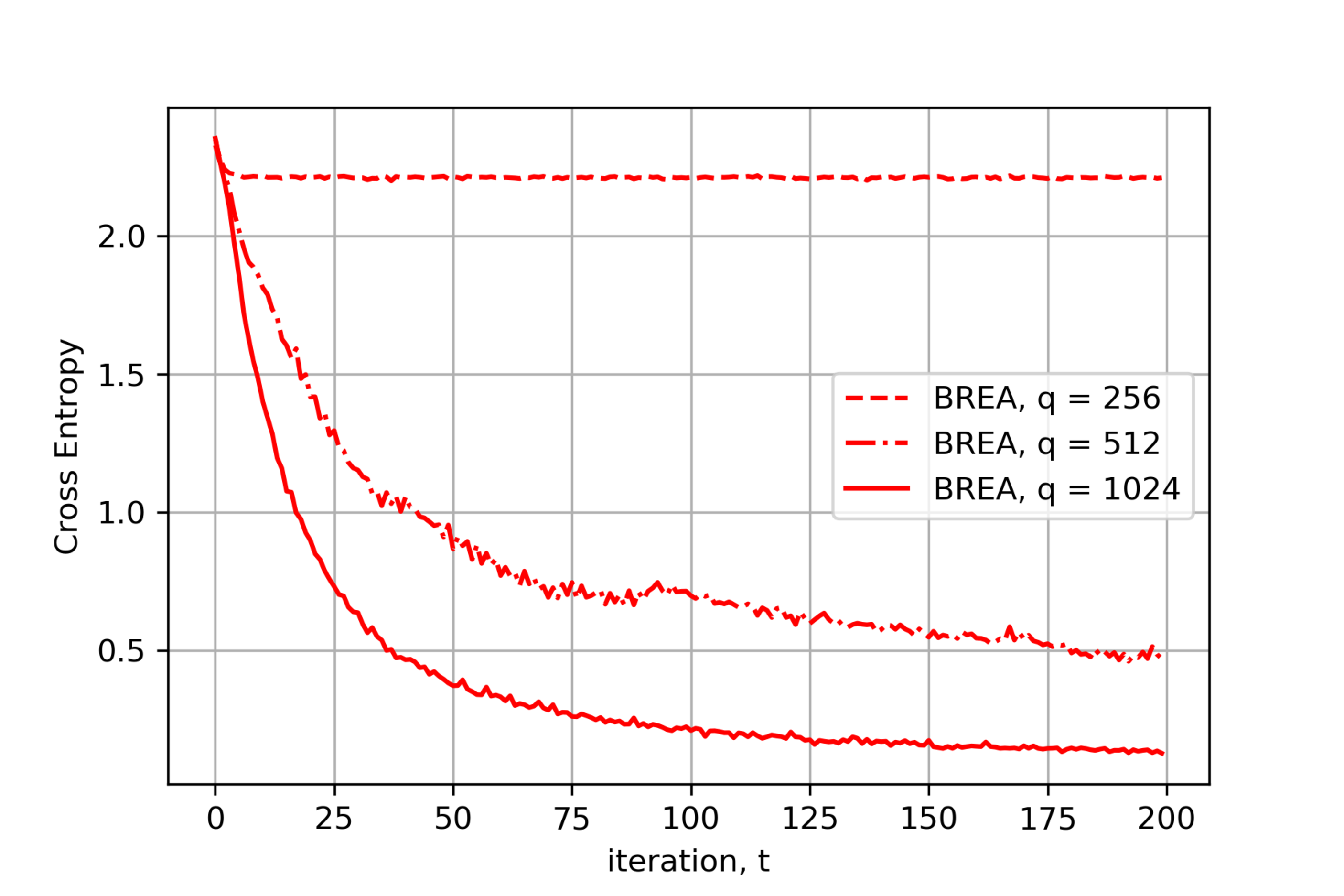}
\caption{Convergence of {\byz} for different values of the  quantization parameter $q$ in~\eqref{eq:sto_round} with $30\%$ Byzantine users {\color{black}and i.i.d. MNIST dataset}.}
\label{fig:cross_entropy_diff_q}
\vspace{-0.3cm}
\end{figure}

\vspace{0.1cm}
\noindent
{\bf Network architecture:}
We consider an image classification task with 10 classes on the MNIST dataset~\cite{lecun2010mnist}, and train a convolutional neural network with 6 layers~\cite{pmlr-v54-mcmahan17a} including two $5\times5$ convolutional layers with stride 1, where the first and the second layers have 32 and 64 channels, respectively, and each is followed by ReLu activation and $2\times2$ max pooling layer.
It also includes a fully connected layer with $1024$ units and ReLu activation followed by a final softmax output layer.

\vspace{0.1cm}
\noindent
{\bf Experiment setup:}
We assume a network of $N=40$ users where $T=7$ users may collude and $A$ users are malicious. We consider two cases for the number of Byzantine users: i) $0\%$ Byzantine users ($A=0$) and ii) $30\%$ Byzantine users ($A=12$). 
{\color{black} For i.i.d. data distribution, $60000$ training samples are shuffled and then partitioned into $N=40$ users each receiving $1500$ samples.}
Honest users utilize the ADAM optimizer~\cite{kingma2014adam} to update the local update by setting the size of the local mini-batch sample to $|\xi_i^{(t)}|=500$ for all $i\in[N],t\in[J]$ where $J$ is the total number of iterations.
Byzantine users generate vectors uniformly at random from $\mathbb{F}_p^{d}$ where we set the field size $p=2^{32} - 5$, which is the largest prime within $32$ bits.
For both schemes, {\byz} and FedAvg, the number of models to be aggregated is set to $m=|\mathcal{S}|=13<N-2A-2$. FedAvg randomly selects $m$ models at each iteration while {\byz} selects $m$ users from~\eqref{eq:mKrum_iter}.

\vspace{0.1cm}
\noindent
{\bf Convergence and robustness against Byzantine users:}
Figure \ref{fig:test_accuracy} shows the test accuracy of {\byz} and FedAvg for different number of Byzantine users. 
We can observe that {\byz} with $0\%$ and $30\%$ Byzantine users is as efficient as FedAvg with $0\%$ Byzantine users, while FedAvg does not tolerate Byzantine users. 
Figure \ref{fig:cross_entropy} presents the cross entropy loss for {\byz} versus FedAvg for different number of Byzantine users. We omit the FedAvg with $30\%$ Byzantine users as it diverges. 
We observe that {\byz} with $30\%$ Byzantine users achieves convergence with comparable rate to FedAvg with $0\%$ Byzantine users, while providing robustness against Byzantine users and being privacy-preserving. For all cases of {\byz} in Figures~\ref{fig:test_accuracy} and \ref{fig:cross_entropy}, we set the quantization value in~\eqref{eq:sto_round} to $q=1024$.
Figure~\ref{fig:cross_entropy_diff_q} further illustrates the cross entropy loss of {\byz} for different values of quantization parameter $q$. We can observe that {\byz} with a larger value of $q$ has better performance because the variance caused by the quantization function $Q_q$ defined in~\eqref{eq:sto_round} gets smaller as $q$ increases. On the other hand, given the field size $p$, the quantization parameter $q$ should be less than a certain threshold to ensure~\eqref{eq:cond2_fieldsize} holds.

{\color{black}
\noindent
{\bf Non-i.i.d dataset:}
To investigate the performance of {\byz} with non-i.i.d dataset, we partition the MNIST data over $N=100$ users by using the same non-i.i.d setting of \cite{pmlr-v54-mcmahan17a}. In this setting, $60000$ training samples are sorted by their} {\color{black} digit labels and partitioned into $200$ shards of size $300$. Then, $N=100$ users are considered, where each user receives two shards with different digit labels.
We assume $T=20$ users may collude and consider two cases for the number of Byzantine users: i) $0\%$ Byzantine users ($A=0$) and ii) $20\%$ Byzantine users ($A=20$). 
Honest users utilize the ADAM optimizer to update the model by setting the size of the local mini-batch sample to $|\xi_i^{(t)}|=10$ and local iterations to $5$ (i.e., each user locally updates the model $5$ times before sending to the server at each iteration).
Byzantine users generate vectors uniformly at random from $\mathbb{F}_p^{d}$.
For both schemes, {\byz} and FedAvg, the number of models to be aggregated is set to $m=50$.
Figure \ref{fig:noniid_test_accuracy} shows the test accuracy of {\byz} and FedAvg for different number of Byzantine users.
Even for the non-i.i.d. dataset, {\byz} with $0\%$ and $20\%$ Byzantine users empirically demonstrates comparable test accuracy to FedAvg with $0\%$ Byzantine users.
}

\begin{figure}[t]
\centering
\includegraphics[width=\linewidth]{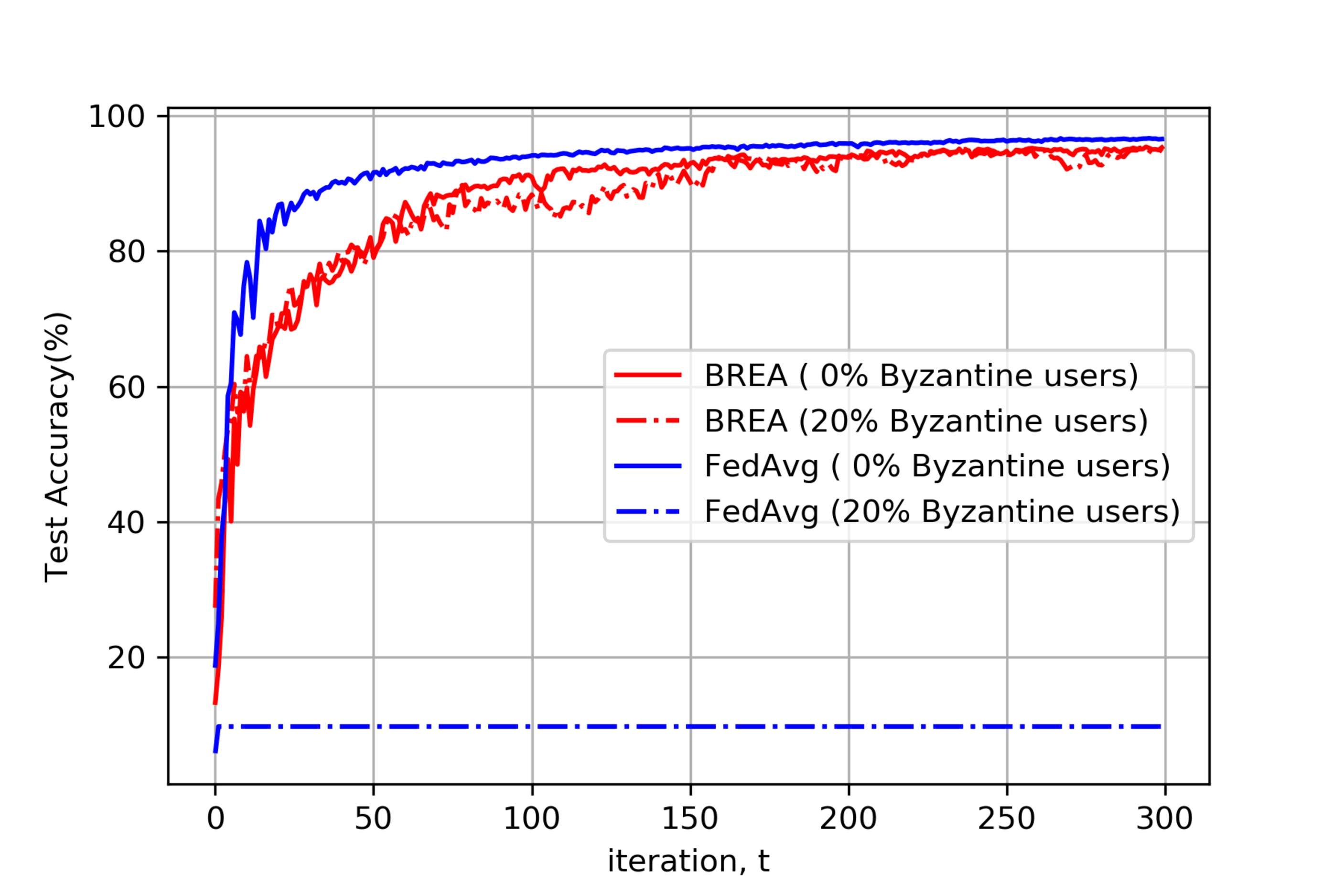}
\vspace{-0.5cm}
\caption{\color{black}Test accuracy of {\byz} and FedAvg~\cite{pmlr-v54-mcmahan17a} for different number of Byzantine users with non-i.i.d MNIST dataset.}
\label{fig:noniid_test_accuracy}
\vspace{-0.3cm}
\end{figure}

\begin{figure}[t]
\centering
\includegraphics[width=\linewidth]{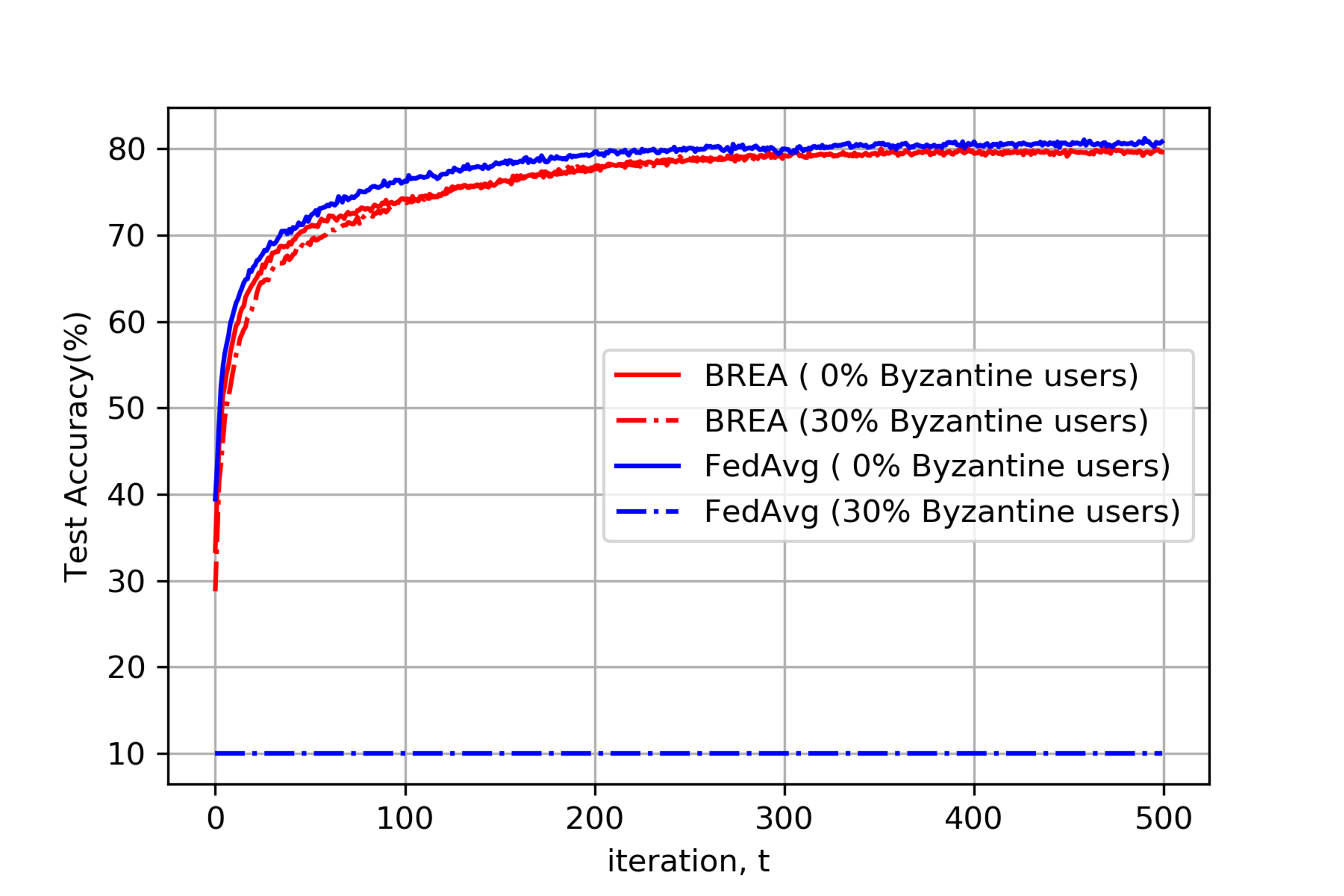}
\vspace{-0.5cm}
\caption{Test accuracy of BREA and FedAvg [1] for different number of Byzantine users with CIFAR-10 dataset.}
\vspace{-0.3cm}
\label{fig:cifar_test_accuracy}
\end{figure}

{\color{black}
\noindent
{\bf CIFAR-10 dataset:}
To investigate the performance of {\byz} with a larger dataset, we additionally experiment with CIFAR-10~\cite{krizhevsky2009learning}, by using the same model architecture of \cite{mcmahan2016communication} (about $10^6$ model parameters).
We use the same setting for parameters $N,A,m,p$ as the setting of i.i.d MNIST dataset.
Figure~\ref{fig:cifar_test_accuracy} shows that the BREA algorithm (with $30\%$ Byzantine users) has comparable test accuracy to the FedAvg algorithm (with no Byzantine users).
}

\section{Conclusion} \label{Sec:conclusion} 
This paper presents the first single-server solution for Byzantine-resilient secure federated learning. 
Our framework is based on a verifiable secure outlier detection strategy to guarantee robustness of the trained model against Byzantine faults, while protecting the privacy of the individual users. We provide the theoretical convergence guarantees and the fundamental performance trade-offs of our framework, in terms of the number of Byzantine adversaries and the user dropouts the system can tolerate. 
In our experiments, we have implemented our system in a distributed network by using two ways of partitioning the MNIST dataset over up to $N=100$ users: i.i.d. and non-i.i.d. data distribution. For both settings, we numerically demonstrated the convergence behaviour while providing robustness against Byzantine users and being privacy-preserving.
Future directions include developing single-server Byzantine-resilient secure learning architectures considering the heterogeneous environments in terms of computation and communication resources, developing efficient communication architectures, and quantization techniques.

In this paper, we utilize the distance based outlier detection approach to guarantee the robustness against Byzantine users. 
There have been many outlier detection approaches for Byzantine resilient SGD methods (or federated learning).
Existing works can be roughly classified into two main categories: 1) geometric median \cite{blanchard2017machine, chen2017distributed} and 2) coordinate-wise median \cite{pmlr-v80-yin18a, alistarh2018byzantine, yang2019byzantine}.
Our framework utilizes the geometric median based approach because it is more efficient in the secure domain. 
In order to preserve the privacy of local updates, the coordinate-wise median based approaches require a  multiparty secure comparison protocol \cite{kerschbaum2009performance} to find the median value of each coordinate. 
On the other hand, in our framework, the server compares the pairwise distances in plaintext, which is much more efficient than the secure comparison protocol.
Extending the coordinate-wise median based algorithms to the secure domain would be a very interesting future direction.


\vspace{-0.25cm}
\section*{Acknowledgement}
\vspace{-0.1cm}
This material is based upon work supported by Defense Advanced Research Projects Agency (DARPA) under Contract No. HR001117C0053, ARO award W911NF1810400, NSF grants CCF-1703575 and CCF-1763673, ONR Award No. N00014-16-1-2189, and research gifts from Intel and Facebook. The views, opinions, and/or findings expressed are those of the author(s) and should not be interpreted as representing the official views or policies of the Department of Defense or the U.S. Government.

\bibliographystyle{IEEEtran}
\bibliography{main.bbl}

\appendix
\subsection{The {\byz} Framework with Local Updating Schemes} \label{App:localupdating}

In this section, we present how {\byz} framework can work with local updating schemes where each user locally takes multiple-steps of stochastic gradient descent using its local data before sending the updated model to server.

The {\byz} framework with local updating schemes follows the same steps from Section \ref{sec:quantization} to Section \ref{sec:secure_aggr} by changing the local update $\mathbf{w}_i^{(t)}$ of user $i$ from a gradient estimate to the updated model by using multiple-steps of stochastic gradient descent (SGD).
Then, the update equation of a benign user $i$ can be expressed as follows,
\begin{align}
    &\mathbf{v}_i^{(t+1)} = \mathbf{w}_i^{(t)}-\gamma_t g(\mathbf{w}_i^{(t)},\mathbf{\xi}_i^{(t)}) \\
    &\mathbf{w}_i^{(t+1)} = 
        \left\{
        \begin{array}{ll}
              \frac{1}{m}\sum_{i\in\mathcal{S}^{(t+1)}}\mathbf{v}_i^{(t+1)} 
              &\text{if } t+1 \in {\mathcal{J}_E}\\ 
              \mathbf{v}_i^{(t+1)} 
              &\text{if } t+1 \notin {\mathcal{J}_E}\\ 
        \end{array}
        \right.
\end{align}
where $\mathcal{J}_E$ denotes the set of global synchronization steps, i.e., $\mathcal{J}_E=\{nE\;|\;n=1,2,\cdots\}$, $E$ is the number of local iterations, and $\mathcal{S}^{(t+1)}$ is the index set of selected users at iteration $t+1$ by utilizing the multi-Krum algorithm described in Section \ref{sec:selection}.
An additional vector $\mathbf{v}_i^{(t+1)}$ is introduced to represent the immediate result of one step SGD update from $\mathbf{w}_i^{(t)}$.
Malicious user $i$ can send any random vector $\mathbf{b}_i^{(t+1)}\in\mathbb{R}^d$, which we represent as $\mathbf{v}_i^{(t+1)}=\mathbf{b}_i^{(t+1)}$.
We note that the distance-based outlier detection algorithm of which inputs are the local models instead of local gradient works effectively as the distance between the local models is equal to the multiplication of the learning rate $\gamma$ and the distance between the local gradients.
Consequently, the aggregation step at the server side can be represented as 
\begin{equation}\label{eq:output_BREA}
    \mathbf{w}^{(t)} = \sum_{i\in\mathcal{S}^{(t)}} \mathbf{w}_i^{(t)} \text{\;\;\;if } t \in {\mathcal{J}_E}.
\end{equation}

\end{document}